\documentclass[10pt,journal,compsoc]{IEEEtran}
\IEEEoverridecommandlockouts

\usepackage{xcolor}
\usepackage{amsthm,amssymb,multirow,amsmath,amsfonts,cite,mathtools}
\usepackage{bm}
\usepackage[ruled,vlined,algo2e]{algorithm2e}
\usepackage{algorithmic}
\usepackage{algorithm}
\usepackage{subcaption}
\usepackage{multirow}
\usepackage{steinmetz}
\usepackage[normalem]{ulem}
\usepackage[colorlinks=false]{hyperref}
\usepackage{tikz}
\usepackage{nicefrac}
\usepackage[cmintegrals]{newtxmath}
\usepackage{ulem}
\usepackage[T1]{fontenc}
\providecommand{\ceil}[1]{\left \lceil #1 \right \rceil }

\usetikzlibrary{arrows}

\newsavebox{\mybox}
\newcommand*\dif{\mathop{}\mathrm{d}}
\newcommand*\ts{T_{\text{s}}}

\newtheorem{lemma}{Lemma}

\newtheorem{prop}{Proposition}

\def\BibTeX{{\rm B\kern-.05em{\sc i\kern-.025em b}\kern-.08em
    T\kern-.1667em\lower.7ex\hbox{E}\kern-.125emX}}

\ifCLASSINFOpdf
\else
\fi

\begin{document}
\title{Energy Efficient Sampling Policies for Edge Computing Feedback Systems}

\author{Vishnu Narayanan Moothedath,~\IEEEmembership{}
        Jaya Prakash Champati,~\IEEEmembership{}
        and~James Gross~\IEEEmembership{}
        
\IEEEcompsocitemizethanks{\IEEEcompsocthanksitem Vishnu Narayanan Moothedath and James Gross are with the Department of Intelligent Systems, KTH Royal Institute of Technology, Malvinas Väg 10, Stockholm 11428, Sweden. E-mail: \{vnmo,jamesgr\}@kth.se
\IEEEcompsocthanksitem Jaya Prakash Champati is with the Edge Networks Group, IMDEA Networks Institute, Avda. del Mar Mediterraneo 22, 28918 Leganes (Madrid), Spain. E-mail: jaya.champati@imdea.org}%
}


\IEEEtitleabstractindextext{%
\begin{abstract}
We study the problem of finding efficient sampling policies in an edge-based feedback system, where sensor samples are offloaded to a back-end server that processes them and generates feedback to a user. Sampling the system at maximum frequency results in the detection of events of interest with minimum delay but incurs higher energy costs due to the communication and processing of redundant samples. On the other hand, lower sampling frequency results in higher delay in detecting the event, thus increasing the idle energy usage and degrading the quality of experience. We quantify this trade-off as a weighted function between the number of samples and the sampling interval. We solve the minimisation problem for exponential and Rayleigh distributions, for the random time to the event of interest. We prove the convexity of the objective functions by using novel techniques, which can be of independent interest elsewhere. We argue that adding an initial offset to the periodic sampling can further reduce the energy consumption and jointly compute the optimum offset and sampling interval. We apply our framework to two practically relevant applications and show energy savings of up to $36\%$ when compared to an existing periodic scheme.
\end{abstract}
}

\maketitle
\IEEEdisplaynontitleabstractindextext
\IEEEpeerreviewmaketitle

\IEEEraisesectionheading{\section{Introduction}\label{sec:intro}}
\IEEEPARstart{W}{ith} the advent of next-generation mobile networks such as 5G Release 15 and 16, there is an increasing interest in realising various real-time services and applications.
Perhaps most prominently, this materialises with the Release 16 features of URLLC (ultra-reliable low latency communication) targeting sub-millisecond end-to-end delays primarily for industrial automation applications.
However, in addition to these extreme use cases, a plethora of new applications are arising that all process states of reality and accurately provide feedback either to devices or humans.
Examples of such feedback systems with low latency requirements are human-in-the-loop applications like augmented reality, wearable cognitive assistants (WCA), or ambient safety.
Also, in the domain of cyber-physical systems (CPS), such applications are prominent, for example, in the context of automated video surveillance or distributed control systems.
All these applications have in common that feedback depends on state capture and timely processing, whereas essential state changes are random events and hence an efficient operation of the application becomes a central aspect of the system.
This is even more emphasised by the recent trend to place most of the processing logic of such feedback systems with edge computing facilities, leveraging supposedly ubiquitous real-time compute capabilities with the additional costs of offloading compute tasks (in terms of communication delays and energy consumption). 

In this paper, we study approaches that enable capturing the relevant system changes in edge-based feedback systems while striking a balance with the total energy consumption. 
We consider feedback systems that monitor a process (or human activities) via sampling, while only reacting to a sub-set of samples, referred to as \textit{essential events}, that lead to a system change -- for instance, a new augmentation towards a human user, an alarm in a surveillance system, or an actuation or fault detection in a general CPS set-up. 
After an essential event is captured and processed (including the generation of feedback), the feedback system transits to the next state where it starts monitoring for an essential next event to happen.
The trade-off that we study relates to the strategy applied to sample the process. 
More frequent sampling leads to a timely capture of the essential event. 
However, it also leads to the capture of unimportant samples of the process, wasting system resources in terms of energy, communication bandwidth, and computing cycles.
We are interested in mathematically characterising this trade-off between detection delay and energy usage and studying the implied consequences in the system design.

\subsection{Related Works}
Some of the earliest ideas on detecting the relevant system changes (or essential events) come from control systems with applications in manufacturing and observable models \cite{Shewhart}. The survey \cite{Gertler} looks at some important works that focus on failure events in complex systems. This is viewed from a statistical perspective in \cite{page} and is later revisited in \cite{veeravalli2012quickest} where the authors look for stochastic changes in the system with an aim of quickest detection. Moving on to the recent research on edge computing and CPS, some important works use data mining and post-processing to detect events from the collected data \cite{dataMining}. We, on the other hand, are interested in detecting live events and generating corresponding feedback to dictate the process. 
Examples of such an event detection occur in real-time video analytics systems which are explored extensively in recent years from multiple perspectives \cite{ananthanarayananreal,edgebox}. 
However, the aspect of energy consumption is not considered in the above works. Instead of using minimal energy, they look only at detecting the relevant system changes as fast as possible. 

Wireless video surveillance, where the video frames are captured by sensors and are sent to the processing node over wireless sensor networks (WSN) is studied in \cite{videoSurveillanceSurvey}. 
The authors discuss the challenges faced by these systems including the energy consumption -- that we are particularly interested in from a general event detection perspective.
The authors discuss some of the adopted methods to reduce energy consumption. These methods include optimising sensor topologies \cite{topology}, optimising video coding and transmission techniques \cite{VidCoding}, and forcing node cooperation between multiple sensors \cite{NodeCooperation}. 
In a similar context of object detection and tracking, \cite{energyefficientobjectdeteciton} discusses the energy saving by sending the camera to an idle state where the frames are dropped for a duration determined adaptively based on the speed of the object. Energy-efficient surveillance and tracking using a network of sensors are also discussed in \cite{smartSleeping}, where only a subset of the sensors are activated at any given time.
A different but widely studied method to save energy is offloading the sensor data \cite{offloading_MEC1,offloading_MEC2}. The drawback of an increased latency compounded on a large number of samples during an offloading is addressed (to a certain extent) by making offloading decisions \cite{survey} for the samples. This includes binary decisions \cite{JayaLiang_SemiOnlineAlgos,JayaLiang_singleRestart,offloading_binary1,offloading_binary2,offloading_binary3,offloading_binary4}, partial offloading decisions \cite{energyAwareOffloading,offloading_partial1,offloading_partial2}, and stochastic decisions \cite{offloading_stochastic}. 

\subsection{Contributions}
In this paper, we focus on detecting the essential events of an edge-based feedback system in real-time in an energy-efficient manner. 
In contrast to all the above-mentioned works, our approach minimises energy consumption by reducing the amount of data generated by the sensors thereby reducing the total amount of data in the communication and processing pipeline. 
While some of the existing works reduce energy by lightly processing and filtering the samples before transmitting, or taking a cluster wise sensing decision, we look at statistically determining the optimum sampling or sensing time instances. 
While the idea of offloading tends to concentrate on the sensor side and take a hit at total energy consumption by shifting the energy usage to the edge device, our work provides a framework for reducing the total energy consumption in the system. 
To the best of our knowledge, our work is the first attempt to find the energy-optimal sampling points of an edge-based feedback system for capturing the relevant system changes.

For illustrating the performance of the proposed design, we focus on two relevant and practical systems that are considered in the related works. 
The first one is a general CPS that aims at fault detection. Such systems are typically characterised by exponentially distributed inter-failure intervals, a relatively small amount of data transfer, and a low-power communication technology. A basic direction on how to approach the above-mentioned trade-off for a CPS is given by the authors in \cite{OptSampling_iccSage}. 
The second system is a video analytics system (VAS). The motivation for considering the VAS comes from a WCA system, where a human task progress is monitored continuously for the detection of the task completion. The task completion time comprises of multiple system delays (communication, processing etc.) and a delay that is tied to the response time and skill of the human user. Previous works on WCA \cite{TowardsWCA} and general distribution fitting indicates that this task completion time can be modelled as a random variable following a Rayleigh distribution. Other characteristics of a VAS are large data transfer and the requirement of a high throughput communication technology, owing to the continuous sequence of video frames that needs to be transmitted to the processing node. 

Our key contributions are listed below.
\begin{itemize}
    \item We pose the problem to find the energy-minimising periodic sampling interval of an edge-based feedback system as an optimisation problem. The convexity of this problem is proved for certain distributions by developing a novel approach that uses the Poisson sum formula and Fourier transforms -- which can be of independent academic interest elsewhere. The problem is solved using a lightweight bisection algorithm that converges exponentially to the optimum. 
    
    \item We prove that adding an initial offset to the periodic sampling further reduces the expected energy consumption when the TTE distribution is not exponential. Thus we pose a more generic optimisation problem in two variables and propose an algorithm to find a solution that achieves near-optimality.
    
    \item Using simulations, we study the energy reduction on systems with a wide range of parameterisations by considering the CPS and the VAS use-cases. Particularly for the VAS, we show the inefficiency of systems used in practice and the larger potential of energy optimisation where we obtain up to $40\%$ increase in battery life. One reason for this is the larger communication content and power requirements for the VAS, which point to the increased relevance of the proposed solution in these future systems.
    We also observe that oversampling needs to be particularly avoided as it introduces a substantial increase in energy usage.
\end{itemize}

The paper is organised as follows. In Section \ref{sec:sysmodel}, we discuss the system model. In Sections \ref{sec:analysis} and \ref{sec:analysis2}, we lay down the general solution approach and solutions for exponential and Rayleigh distributed TTEs. We discuss the numerical results in Section \ref{sec:numerical} and conclude in Section \ref{sec:conclusion}.
\section{System Model and Problem Statement}\label{sec:sysmodel}
Consider a feedback system consisting of a mobile terminal (simply terminal) and a back-end server (simply back-end), that is designed to monitor a process through sampling. The terminal captures samples that are sent to the back-end for processing. Immediately after the occurrence of an essential event (simply event), the process moves to an intermediate state where no more events are expected. The next sample drawn at or after this transition point -- referred to as a \textit{successful sample} -- indicates the event detection at the back-end's processor, and causes feedback to the terminal. The reception of this feedback triggers the start of a fresh \textit{monitoring cycle} to detect the next event. Within a monitoring cycle, the feedback is generated only to the successful sample, and all other samples are discarded by the back-end. The time taken from the start of a monitoring cycle to the event is termed as \textit{Time to event (TTE)}, and the time between this event and the corresponding feedback is termed as \textit{Time to feedback (TTF)}. 
We denote the TTE using the random variable $\mathcal{T}$, and a value for $\mathcal{T}$ is denoted by $t$.
The timing diagram of the system is given in Fig. \ref{fig:model}. Here we focus on modelling and designing the monitoring cycle corresponding to a single event, and the proposed design can be independently extended to all the (potentially predetermined set of) events that the system is supposed to detect.
\begin{figure}[ht]
\centerline{\includegraphics[width=\linewidth]{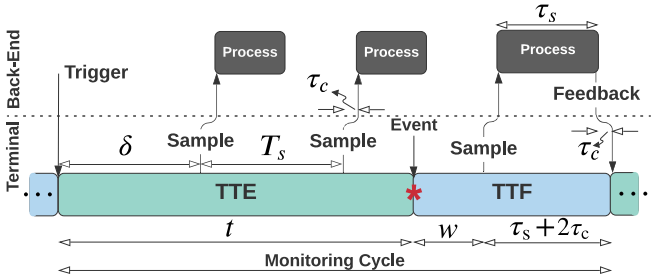}}
\caption{Timing diagram of an arbitrary monitoring cycle.}
\label{fig:model}
\end{figure}

For example, in a CPS, the event might correspond to the detection of a system failure. The feedback potentially triggers a reset and the terminal starts monitoring for the next failure. In the VAS (or WCA) system studied in \cite{ChenZhuoHu,munoz2020impact}, the user is assigned to complete a set of tasks and the events correspond to the completion of each task. The terminal takes snapshots (video frames) of the user activity using a camera and sends them to the back-end for image processing. 

The terminal can sample the process, transmit it to the back-end and can receive feedback whenever available. The back-end on the other hand is capable of receiving the sample, processing it and transmitting the feedback after an essential event detection. At all other times, the terminal and back-end go to their respective idle mode to save energy, and the power consumed while in this mode is denoted by $P_0$. 
Though the back-end can possibly serve and manage multiple terminals, we restrict our study to the back-end's interaction with a single terminal. 

Sampling the system for detecting the event is governed by a set of sampling policies ${\Pi}$. We consider sampling policies $\pi\in\Pi$ that sample the system periodically with a sampling interval $\ts$, except for the first sample, which is sampled after waiting for a duration $\delta \geq \ts$. We refer to $\delta$ as \textit{offset}. This offset takes care of a minimum time threshold, sampling before which only adds to the energy wastage. We elaborate this in detail in Section \ref{sec:analysis2}.
Apart from the set of sampling policies $\Pi$, we also consider a special case with $\delta=\ts$, referred to as $\hat{\Pi}$. In other words, under $\Pi$ the samples are taken at $t=k\ts +\delta,\;k\in\mathbb{N}$ whereas under $\hat{\Pi}$ the samples are taken at $t=k\ts ,\;k\in\mathbb{N}^+$.
While the optimum policy under $\Pi$ comprises of an optimum sampling interval $\ts ^*$ and an optimum phase $\delta^*$, the optimum policy under the analytically simpler subset $\hat{\Pi}$ corresponds to (only) an optimum sampling interval $\ts^\#$. We use the notation $(\cdot)^\#$ instead of $(\cdot)^*$ to have a clear distinction between the underlying set of policies. 
Following the same conventions, we refer to these optimum policies as $\pi^*$ and $\pi^\#$, respectively.

The number of samples taken is denoted by a random variable $\mathcal{S}$, and a value for $\mathcal{S}$ is denoted by $s$. This includes both the discarded samples taken during TTE as well as the final sample that leads to detecting the event. 
This TTF consists of a (potential) random lag $\mathcal{W}$ until the next sample is drawn as well as a deterministic processing and two-way communication delays. The random lag $\mathcal{W}$ is referred to as \textit{wait} and $w$ denotes a value of $\mathcal{W}$. Note that only the processing and communication delay corresponding to the final sample contributes to the delay after the wait. In this work, we assume that the total power consumption during transmission and reception is the same at both the terminal and the back-end and we refer to it simply as communication power denoted by $P_{\text{c}}$, which is typically much larger than the idle power $P_0$. 
We also assume that the communication delay in either direction is $\tau_{\text{c}}$ and the processing time of the successful sample be $\tau_\text{s}$.

When not performing a transmission, reception, or processing, both the terminal and the back-end enter idle mode, which incurs an idle power consumption. The total time within a state, during which the terminal or back-end is in an idle mode is referred to as the \textit{idle time} and is denoted by the random variable $\mathcal{T}_0$. 
The notations used and their meanings are reiterated in TABLE \ref{tab:notations} for readability. 
\begin{table}[htb]
\centering
    \begin{tabular}{|llll|}
    \hline
    $\mathcal{S}$&number of samples&
    $\mathcal{W}$&wait time\\
    
    $\mathcal{T}$& time to event (TTE)&
    $\mathcal{T}_0$&idle time\\
    
    ${F}_\mathcal{X}(\cdot)$&CDF $X$&
    $\bar{F}_\mathcal{X}(\cdot)$& CCDF of $\mathcal{X}$\\
    
    $\ts $&sampling interval&
    $\delta$&offset\\
    
    $\tau_{\text{c}} $&communication delay&
    $\tau_{\text{s}} $&processing delay\\

    $P_{\text{c}} $&communication power&
    $P_0$&idle power\\
    
    $\mathrm{E}$&energy&
    $\mathcal{E}$&energy penalty\\
    
    $\mathbb{N}$&\{0,1,2,\dots\}&
    $\mathbb{N}^+$&\{1,2,3,\dots\}\\
    
    \hline
    \end{tabular}
    \caption{Table of notations.}
    \label{tab:notations}
\end{table}

\subsection{Problem Statement}
An ideal sampling policy should sample the system immediately after the event so that the wait is zero and the number of samples required per event is exactly 1.
However, such a sampling policy is unattainable given the fact the TTE is randomly distributed. 
As a result, we have to settle with a sampling policy that simultaneously reduces the expected wait and expected number of discarded samples to yield the maximum attainable benefit. Note that this is not straightforward as the wait time and the number of samples shows opposite behaviour with a change in sampling frequency. For instance, an aggressive sampling reduces the wait time, but it also increases the number of samples. In this work, we use energy as a metric to quantify this opposing behaviour. Each sample warrants energy in terms of communication and processing; and during the wait $w$, energy is expended as governed by the idle power. It is this optimal trade-off of energy usage between the sampling frequency and the number of samples that we seek to quantify.

Note that the fundamental random variable here is the TTE $\mathcal{T}$ and the other random variables -- $\mathcal{S}, \mathcal{W}$ and $\mathcal{T}_0$ -- are derived from $\mathcal{T}$ through the selection of the optimization variables $\ts$ and $\delta$. We can compute the idle time denoted by $\mathcal{T}_0$ in terms of other parameters as follows:
\begin{alignat}{2}
    \mathcal{T}_0&&\ =\ &\mathcal{T}+\mathcal{W}+\tau_{\text{s}} -(\mathcal{S}-1)\tau_{\text{c}}\;. \label{eq:idleTime_Terminal}
\end{alignat}
Let $\mathrm{E}$ be the energy required for detecting one event.
\begin{alignat}{2}
    \mathrm{E}&=(\mathcal{S}+1)\tau_{\text{c}} P_{\text{c}} +\mathcal{T}_0P_0\nonumber\\
    &=\mathcal{S}\tau_{\text{c}} (P_{\text{c}} -P_0)+\mathcal{W}P_0+(\mathcal{T}+\tau_{\text{c}} +\tau_{\text{s}} ) P_0+\tau_{\text{c}} P_{\text{c}}\;. \label{eq:Energy_Terminal}
\end{alignat}\par

We aim to minimise the expected energy for a given TTE statistics. 
In \eqref{eq:Energy_Terminal}, the terms except those containing
the number of samples
$\mathcal{S}$ or
the wait
$\mathcal{W}$
are either constants or have constant expectations for a fixed distribution of $\mathcal{T}$. Hence these terms are irrelevant in the optimisation where we minimise the expected energy. Let $\mathcal{E}(\ts ,\delta)$ be the component of the total energy which is relevant for the optimisation and let us call it \textit{energy penalty}. We obtain
\begin{alignat}{1}
\mathcal{E}(\ts ,\delta)=\alpha\mathbb{E}[\mathcal{S}]+\beta\mathbb{E}[\mathcal{W}]\;,\label{eq:epsilon_terminal}
\intertext{where we use the constants $\alpha$ and $\beta$ for mathematical tractability in the upcoming sections. Here,}
\alpha=\tau_{\text{c}} (P_{\text{c}} -P_0) \;\text{ and }\;\beta=P_0\;.\label{eq:alphaBeta}
\end{alignat}
In \eqref{eq:epsilon_terminal}, $\beta\,\mathbb{E}[\mathcal{W}]$ corresponds to the additional energy expended for waiting, and $\alpha$ represents the energy wasted per discarded sample due to the additional communication and processing. The minimum values of $\mathcal{E}$ under the policies $\Pi$ and $\hat{\Pi}$ are referred to as $\mathcal{E}^*$ and $\mathcal{E}^\#$, respectively. That is
 \begin{equation*}
     \mathcal{E}^*=\mathcal{E}(\ts^*,\delta^*) \text{ and } {\mathcal{E}}^\#=\mathcal{E}(\ts^\#,\ts^\#)\;.
 \end{equation*}
 
In what follows, we study the general optimisation problems $\mathcal{P}$ and $\hat{\mathcal{P}}$ under the set of sampling policies ${\Pi}$ and $\hat{\Pi}$, respectively. These optimisation problems are defined as follows:
\begin{alignat}{2}
\begin{split}
\mathcal{P}:&\;\;\pi^*=
    \{\ts ^*,\delta^*\}=\;\underset{\{\ts ,\delta\}}{\text{arg\,min}}\ \mathcal{E}(\ts,\delta)\\
\hat{\mathcal{P}}:&\;\;\pi^\#=
    \{\ts^\#\}=\;\underset{\ts }{\text{arg\,min}}\ \mathcal{E}(\ts,\ts)\;,
\end{split}\label{eq:Ts_optimum}
\end{alignat}
 where $\pi^*$ and $\pi^\#$ are the optimum policies for $\mathcal{P}$ and $\hat{\mathcal{P}}$, respectively.
 We will solve these problems for the CPS and the VAS where the TTEs follow exponential and Rayleigh distributions, respectively.
\section{Solution to the Problem $\hat{\mathcal{P}}$}\label{sec:analysis}
Being the simpler one, we start with the optimisation problem $\hat{\mathcal{P}}$ in this section and find the optimum periodic sampling interval by minimising $\mathcal{E}(\ts,\delta)$ given in \eqref{eq:epsilon_terminal}. Throughout this section, we will drop the redundant second argument $\delta\,(=\!\ts \text{ under }\hat{\mathcal{P}})$ from the energy penalty for simplicity. As the solution is specific to the distribution of $\mathcal{T}$, we will first lay down the solution approach to find the energy penalty for any general distribution and later apply this to particular distributions.

\subsection{General TTE Distribution}
Here, we derive the expressions for the expected number of samples and the expected wait time under a general TTE distribution, which together constitutes $\mathcal{E}(\ts )$. 
\subsubsection{Expected number of samples}\label{subsubsec:general_OP}
Recall that $\mathcal{S}$ is the random number of samples taken for detecting an event. We have
\begin{alignat}{2}
    \mathbb{P}( \mathcal{S}=k)&&\ =\ &\mathbb{P}\big(\ceil{\nicefrac{t}{\ts }}=k\big),\;\forall k\geq 1\nonumber\\
    &&\ =\ & \mathbb{P}(k-1< \nicefrac{t}{\ts }\leq k)\nonumber\\
    &&\ =\ & F_{\mathcal{T}}\big(k\ts \big)-F_{\mathcal{T}}\big((k-1)\ts \big)\nonumber\\
    \Rightarrow \mathbb{E}[\mathcal{S}]&&\ =\ &\sum_{k=1}^{\infty}k\Big(F_{\mathcal{T}}\big(k\ts \big)-F_{\mathcal{T}}\big((k-1)\ts \big)\Big)\;.\label{eq:EOP_general}
\end{alignat}
\subsubsection{Expected wait time}\label{subsubsec:general_WP}
Recall that w denotes a value of the random wait time $\mathcal{W}$. 
We have a fixed set of sampling instances governed by $\ts $. The CDF of the wait $F_\mathcal{W}(w)$ can be obtained by taking the probability of the TTE to fall at most $w$ short of any sampling instance. Even though the TTEs are finite in practice, we consider they can be arbitrarily large for the sake of generalised analysis. As a result, a successful sample can be located anywhere from the first sampling instance to possibly infinity. For real systems, however, there is an upper bound for the TTE (preemption or otherwise) and it implies that the corresponding probability is zero beyond this point. We can compute the CDF of $\mathcal{W}$ as follows.
\begin{alignat}{2}
F_\mathcal{W}(w)&&\ =\ &\sum_{k=1}^{\infty}\mathbb{P}\big(k\ts -w<t\leq k\ts \big)\nonumber\\
&&\ =\ &\sum_{k=1}^{\infty}\big(F_{\mathcal{T}}(k\ts )-F_{\mathcal{T}}(k\ts -w)\big)\;.\label{eq:CDF_w_general}
\intertext{Since $\mathcal{W}$ is a non-negative random variable, we have}
\mathbb{E}[\mathcal{W}]&&\ =\ &{\int_{0}^{\infty}}\big(1-F_\mathcal{W}(w)\big)\,\mathrm{d}w\nonumber\\
&&\ =\ & \int_{0}^{\ts }\Big(1-\sum_{k=1}^{\infty}\big(F_{\mathcal{T}}(k\ts )-F_{\mathcal{T}}(k\ts -w)\big)\Big)\,\mathrm{d}w\;. \nonumber
\intertext{Here, integral limits and infinite sum are finite, and the summand is non-negative. Thus}
\mathbb{E}[\mathcal{W}]&&\ =\ &\ts -\sum_{k=1}^{\infty}\int_{0}^{\ts }\big(F_{\mathcal{T}}(k\ts )-F_{\mathcal{T}}(k\ts -w)\big)\,\mathrm{d}w\;.\label{eq:Ewp_general}
\end{alignat}
\par

Using \eqref{eq:EOP_general} and \eqref{eq:Ewp_general}, we can calculate the energy penalty given in \eqref{eq:epsilon_terminal}.
Clearly, the expected number of samples decrease with an increase in $\ts $. We can see that the summation containing the integral in the expected wait time expression also decrease with an increase in $\ts $, thus increasing $\mathbb{E}[\mathcal{W}]$. This opposing behaviour of the two penalties result in a minima in their weighted sum $\mathcal{E}(\ts )$ (with weights $\alpha$ and $\beta$) at the optimum sampling interval $\ts ^{\#}$. Depending on the distribution of $\mathcal{T}$ and the resultant energy penalty, $\ts ^{\#}$ can be computed using known optimisation techniques or numerical solvers \cite{boyd}. 
\subsection{Exponentially Distributed TTE}
In this subsection, we look into a feedback system where the TTEs are exponentially distributed. This part of the work describing about the solution of $\hat{\mathcal{P}}$ under an exponentially distributed TTE was published earlier by the authors in \cite{OptSampling_iccSage}. 
\begin{lemma}\label{lemma_totPenalty_exponential}
For exponentially distributed TTE with mean $\nicefrac{1}{\lambda}$, the energy penalty $\mathcal{E} (\ts )$ imparted by a periodic sampling policy with a period $\ts $ is given by 
\begin{equation}
    \mathcal{E} (\ts )=\dfrac{\alpha\lambda+\beta\big(e^{-\lambda \ts }+\lambda \ts -1\big)}{\lambda (1-e^{-\lambda \ts })}\;.\label{eq:totPen_exponential}
\end{equation}
\end{lemma}
\begin{proof}
Given that the TTEs are exponentially distributed with rate parameter $\lambda$, we have $F_{\mathcal{T}}(t)=1-e^{-\lambda t}\;$. From \eqref{eq:EOP_general}, the expected number of samples can be computed as
\begin{alignat}{2}
\mathbb{E}[\mathcal{S}]&&\ =\ &\dfrac{1}{1-e^{-\lambda \ts }}.\label{eq:Eop_exponential}
\intertext{Substituting $F_{\mathcal{T}}(t)$ in  \eqref{eq:CDF_w_general}}
F_\mathcal{W}(w)&&\ =\ &\sum_{k=1}^{\infty}\big((1-e^{-\lambda k\ts })-(1-e^{-\lambda (k\ts -w)})\big)\nonumber\\
&&\ =\ &\sum_{k=1}^{\infty}e^{-\lambda k\ts }(e^{\lambda w}-1)=\dfrac{e^{\lambda w}-1}{e^{\lambda \ts }-1}\nonumber\\
\Rightarrow f_\mathcal{W}(w)&&\ =\ &\dfrac{\lambda e^{\lambda w}}{e^{\lambda \ts }-1}\nonumber\\
\Rightarrow \mathbb{E}[\mathcal{W}]&&\ =\ &\dfrac{e^{-\lambda \ts }+\lambda \ts -1}{\lambda (1-e^{-\lambda \ts })}\;.\label{eq:Ewp_exponential}
\intertext{Substituting \eqref{eq:Eop_exponential} and \eqref{eq:Ewp_exponential} in \eqref{eq:Ts_optimum}, we get the energy penalty $\mathcal{E} (\ts )$ for using a particular sampling interval $\ts $ as}
\mathcal{E} (\ts )&&\ =\ &\dfrac{\alpha\lambda+\beta\big(e^{-\lambda \ts }+\lambda \ts -1\big)}{\lambda (1-e^{-\lambda \ts })}\;.\nonumber
\end{alignat}
\end{proof}
\begin{lemma}\label{lemma_totPenaltyConvexity_exponential}
The energy penalty $\mathcal{E} (\ts )$
is convex in $\ts\;$.
\end{lemma}
\begin{proof}
In the following, we drop the argument $\ts$ from $\mathcal{E} (\ts )$ for simplicity in presentation. Note from \eqref{eq:alphaBeta} that $\beta>0$.
\begin{alignat}{3}
     \mathcal{E}'\ :=\ &\dfrac{\dif\mathcal{E}}{\dif \ts }\ =\ &&\dfrac{\beta\big(1-e^{-\lambda \ts }(\lambda \ts +\tfrac{\alpha}{\beta}\lambda+1)\big)}{(1-e^{-\lambda \ts })^2}\;,\nonumber\\
     \mathcal{E}''\ :=\ &\dfrac{\dif ^2\mathcal{E}}{\dif \ts ^2}\ =\ &&\dfrac{\beta\lambda e^{-\lambda \ts }\big((1+e^{-\lambda \ts })(\lambda \ts +\tfrac{\alpha}{\beta}\lambda)+2e^{-\lambda \ts }-2 \big)}{(1-e^{-\lambda \ts })^3}\;.\nonumber
\end{alignat}
\begin{alignat}{2}
\mathcal{E}''\ \geq 0 \Rightarrow\ &\Tilde{\mathcal{E}}\ &:=\ &(1+e^{-\lambda \ts })(\lambda \ts +\tfrac{\alpha}{\beta}\lambda)+2e^{-\lambda \ts }-2 \geq 0\;,\nonumber\\
&\Tilde{\mathcal{E}}'\ &:=\ &\dfrac{\dif \Tilde{\mathcal{E}}}{\dif \ts }\ =\ \lambda(1-e^{-\lambda \ts }(\lambda \ts +\tfrac{\alpha}{\beta}\lambda+1))\;,\nonumber\\
&\Tilde{\mathcal{E}}''\ &:=\ &\dfrac{\dif ^2\Tilde{\mathcal{E}}}{\dif \ts ^2}\ =\ \lambda^2e^{-\lambda \ts }(\lambda \ts +\tfrac{\alpha}{\beta}\lambda)\;.\nonumber
\end{alignat}
From the above expressions, since $\Tilde{\mathcal{E}}''\geq0\ \forall \ts $, we can conclude that $\Tilde{\mathcal{E}}$ is globally convex and any infimum point is its minimum. To find this infimum:
\begin{align*}
 \Tilde{\mathcal{E}}'=0&\Rightarrow \lambda \ts +\tfrac{\alpha}{\beta}\lambda= e^{\lambda \ts }-1\;.
\end{align*}
Substituting in the above expression for $\Tilde{\mathcal{E}}$, we obtain 
\begin{align*}
\underset{\ts >0}{\min}\{\Tilde{\mathcal{E}}\} =& (1+e^{-\lambda \ts })(e^{\lambda \ts }-1)+2e^{-\lambda \ts }-2\\
=&e^{\lambda \ts }+e^{-\lambda \ts }-2\\
\geq&0\ \forall \ts >0\\
    \Rightarrow \Tilde{\mathcal{E}}\geq& 0\ \forall \ts >0\\
    \Rightarrow \mathcal{E}''\geq& 0\ \forall \ts >0\;.
\end{align*}
As $\ts $ is non-negative, the energy penalty $\mathcal{E}$ is convex.
\end{proof}
\begin{prop}\label{prop:exponentialSolution}
The optimum sampling interval  $\ts ^{\#}$ under an exponentially distributed TTE with mean $\nicefrac{1}{\lambda}$ is the solution to the expression 
\begin{equation}
e^{\lambda \ts ^{\#}}-\lambda \ts ^{\#} =\tfrac{\alpha}{\beta}\lambda+1\;.\label{lemma_Ts*_exponential}
\end{equation}
\end{prop}
\begin{proof}
The proof follows from Lemma \ref{lemma_totPenalty_exponential} and Lemma \ref{lemma_totPenaltyConvexity_exponential}. We can find the optimum by equating the first derivative of total penalty to zero.
 \begin{align}
     \dfrac{\dif \mathcal{E} (\ts )}{\dif \ts }=0 &\Rightarrow 1-e^{-\lambda \ts }(\lambda \ts +\tfrac{\alpha}{\beta}\lambda+1)=0\nonumber\\
     \Rightarrow e^{\lambda \ts }-\lambda \ts  &=\tfrac{\alpha}{\beta}\lambda+1\;.\nonumber
 \end{align}
 
 \end{proof}
We know that $e^x-x$ is a monotonically increasing convex function in $x$ with $e^x-x\geq 1,\ \forall x\geq 0$. Hence, this single variable expression can be solved using well-known numerical solvers. \par

Though the value of energy penalty depends on the values of $\beta$ and $\alpha$, the optimum sampling interval $\ts ^{\#}$ only depends on their ratio. Furthermore, from \eqref{eq:alphaBeta}, it can be seen that this ratio $\nicefrac{\beta}{\alpha}$ does not depend on the individual power figures but only on the percentage additional power necessary for communication or processing when compared to their respective idle power requirement. In other words, for fixed $\lambda$, $\nicefrac{\beta}{\alpha}$ and thus the optimum sampling interval is a function of only $\tau_{\text{c}} $ and $\nicefrac{P_{\text{c}} }{P_0}$.

\subsection{Rayleigh Distributed TTE}
We will now consider a Rayleigh distributed TTE with mean $\mu\!=\!\sigma\sqrt{\tfrac{\pi}{2}}$ and CDF
$F_\mathcal{T}(t)\!=\!1-e^{-t^2/2\sigma ^2}$; where $\sigma$ is the scale parameter. 

\begin{lemma}
The expected number of samples is given by
\begin{equation}
\mathbb{E}[\mathcal{S}] = \sum_{k=0}^{\infty}e^{-k^2\ts ^2/2\sigma^2}\;.\label{ExpSamples_Rayleigh}
\end{equation}
\end{lemma}
\begin{proof}
Substituting the Rayleigh CDF in \eqref{eq:EOP_general} gives
\begin{alignat}{2}
\mathbb{E}[\mathcal{S}]&&\ =\ &\sum_{k=1}^{\infty}k\bigg(e^{-(k-1)^2\ts ^2/2\sigma^2}-e^{-k^2\ts ^2/2\sigma^2}\bigg)\;.\nonumber
\end{alignat}
Since the positive and negative terms converges individually to a finite value, we can rearrange the terms to complete the proof.
\end{proof}

\begin{lemma}
The expected wait time is given by
\begin{equation}
\mathbb{E}[\mathcal{W}] = \ts  \sum_{k=0}^{\infty}e^{-(k\ts )^2/2\sigma^2}-\sigma\sqrt{\tfrac{\pi}{2}}\;.\label{ExpWait_Rayleigh}
\end{equation}
\end{lemma}
\begin{proof}
Substituting the Rayleigh CDF in \eqref{eq:Ewp_general} gives
\begin{alignat}{2}
\mathbb{E}[\mathcal{W}]&&\ =\ &\ts -\sum_{k=1}^{\infty}\int_0^{\ts }\big(e^{\nicefrac{-(k\ts -w)^2}{2\sigma^2}}-e^{\nicefrac{-(k\ts )^2}{2\sigma^2}}\big)\,\mathrm{d}w\nonumber\\
&&\ =\ &\ts+\ts \sum_{k=1}^{\infty}e^{-(k\ts )^2/2\sigma^2}-\sum_{k=1}^{\infty}\int_0^{\ts }e^{\nicefrac{-(k\ts -w)^2}{2\sigma^2}}\,\mathrm{d}w\nonumber\\
&&\ =\ &\ts  \sum_{k=0}^{\infty}e^{-(k\ts )^2/2\sigma^2}-\sum_{k=1}^{\infty}\int_0^{\ts }e^{\nicefrac{-(k\ts -w)^2}{2\sigma^2}}\,\mathrm{d}w\;.\nonumber
\end{alignat}
Let $\text{\textit{erf}}(x)$ be the error function defined as $$\text{\textit{erf}}(x)=\dfrac{2}{\sqrt{\pi}}\int_{0}^{x}e^{-t^2}\,\dif t\;.$$
Also, let $\text{\textit{erfc}}(x)=1-\text{\textit{erf}}(x)$ be the complimentary error function. Substituting $\dfrac{(k\ts -w)}{\sigma\sqrt{2}}=t$, we get
\begin{alignat}{2}
&&\ &\mathbb{E}[\mathcal{W}]=\ts  \sum_{k=0}^{\infty}e^{-(k\ts )^2/2\sigma^2}-\sigma\sqrt{2}\sum_{k=1}^{\infty}{\int_{\nicefrac{(k-1)\ts }{\sigma\sqrt{2}}}^{\nicefrac{k\ts }{\sigma\sqrt{2}}}}e^{-t^2}\,\dif t\nonumber\\
&&\ =\ &\ts  \sum_{k=0}^{\infty}e^{-(k\ts )^2/2\sigma^2}-\sigma\sqrt{\tfrac{\pi}{2}}\sum_{k=1}^{\infty}\Big(\text{\textit{erf}}\big({\tfrac{k\ts }{\sigma\sqrt{2}}}\big)-\text{\textit{erf}}\big(\tfrac{(k-1)\ts }{\sigma\sqrt{2}}\big)\Big)\nonumber\\
&&\ =\ &\ts  \sum_{k=0}^{\infty}e^{-(k\ts )^2/2\sigma^2}-\sigma\sqrt{\tfrac{\pi}{2}}\sum_{k=1}^{\infty}\Big(\text{\textit{erfc}}\big(\tfrac{(k-1)\ts }{\sigma\sqrt{2}}\big)-\text{\textit{erfc}}\big({\tfrac{k\ts }{\sigma\sqrt{2}}}\big)\Big)\nonumber\\
&&\ =\ &\ts  \sum_{k=0}^{\infty}e^{-(k\ts )^2/2\sigma^2}-\sigma\sqrt{\tfrac{\pi}{2}}\Big(\text{\textit{erfc}}(0)-\text{\textit{erfc}}(\infty)\Big)\nonumber\\
&&\ =\ &\ts  \sum_{k=0}^{\infty}e^{-(k\ts )^2/2\sigma^2}-\sigma\sqrt{\tfrac{\pi}{2}}\;.\nonumber
\end{alignat}
\end{proof}
Substituting \eqref{ExpWait_Rayleigh} and \eqref{ExpSamples_Rayleigh} in \eqref{eq:Ts_optimum}, we obtain the total energy penalty, given by
\begin{alignat}{2}
    \mathcal{E}(\ts )&&\ =\ &\alpha \sum_{k=0}^{\infty}e^{-k^2\ts ^2/2\sigma^2} +\beta\Big(\ts  \sum_{k=0}^{\infty}e^{-(k\ts )^2/2\sigma^2}-\sigma\sqrt{\tfrac{\pi}{2}}\Big)\nonumber\\
     &&\ =\ & (\alpha+\beta \ts )\sum_{k=0}^{\infty}e^{-k^2\ts ^2/2\sigma^2} - \beta\sigma\sqrt{\tfrac{\pi}{2}}\;.\label{eq:energyPen_rayleigh}
\end{alignat}
Note that this energy penalty can be expressed using the Jacobi third theta function $\theta_3(z,q)$ given by \cite{whittaker_watson_1996} \begin{equation}
    \theta_3(z,q)=\sum_{n=-\infty}^{\infty}q^{n^2}e^{2niz}\;.\label{eq:jacobi}
\end{equation}
Although the Jacobi theta functions do not have a closed-form solution, it is finite and can be computed numerically or from tables. If required by the numerical solvers, this infinite sum can be approximated to a finite sum with a sufficiently large number of terms.

In what follows, we prove the convexity of $\mathcal{E}(\ts)$. Note that it involves a sum of infinite terms, and one can show that some of those terms are non-convex. Therefore, using typical methods to prove convexity is not applicable. Instead, we use a novel approach that involves a transformation of the infinite sum using the Poisson sum formula, which can be of independent interest in proving convexity for a sum of infinite terms, in general.

\begin{lemma}\label{Lemma:PoissonSum}
Let $f(k,t)=e^{-k^2t^2/2\sigma^2}$. Then for $t\in(0,\infty)$,
$$\sum\limits_{k=0}^{\infty}f(k,t)=\sqrt{2\pi}\>\dfrac{\sigma}{t}\sum\limits_{k=0}^{\infty}f(k,\tfrac{2\pi \sigma^2}{t})\;.$$
\end{lemma}
\begin{proof}
According to the Poisson sum formula,
periodic summation of an aperiodic function is equal to the periodic summation of its Fourier transform. That is, for the given function $f(k,t)$
\begin{alignat}{2}
       \sum\limits_{k=-\infty}^{\infty}f(k,t)&=\sum\limits_{n=-\infty}^{\infty}\int_{-\infty}^{\infty}f(x,t)e^{-2\pi inx}\,\dif x\nonumber\\
      &=\sum\limits_{n=-\infty}^{\infty}\int_{-\infty}^{\infty}e^{-x^2t^2/2 \sigma^2}e^{-2\pi inx}\,\dif x\nonumber\\
         &=\sqrt{2\pi}\>\dfrac{\sigma}{t}\sum\limits_{n=-\infty}^{\infty}e^{-2\pi^2n^2 \sigma^2/t^2}\nonumber\\
         \Rightarrow\sum\limits_{k=-\infty}^{\infty}f(k,t)&=\sqrt{2\pi}\>\dfrac{\sigma}{t}\sum\limits_{k=-\infty}^{\infty}f(k,\tfrac{2\pi \sigma^2}{t})\;.\nonumber
         \intertext{Since $f(k,t)$ is an even function of $k$, we get}
         \sum\limits_{k=0}^{\infty}f(k,t)&=\sqrt{2\pi}\> \dfrac{\sigma}{t}\sum\limits_{k=0}^{\infty}f(k,\tfrac{2\pi \sigma^2}{t})\;.\nonumber
         \end{alignat}
\end{proof}

\begin{lemma}\label{Lemma:ConvexityOfEs}
The function $f(t)=\sum_{k=0}^{\infty}f(k,t)$ is convex in $t\;$.
\end{lemma}
\begin{proof}We prove the convexity using the second derivative test. Since the function sequence converges, it is sufficient to prove that the terms inside the summation are independently convex. For this, we make use of the function and its alternate expression given by Lemma \ref{Lemma:PoissonSum}.
\begin{alignat}{2}
f(t)&=1+\sum\limits_{k=1}^{\infty}e^{-k^2t^2/2\sigma^2}\nonumber\\
\Rightarrow f''(t)&=\sum\limits_{k=1}^{\infty}\dfrac{k^2}{\sigma^4}(k^2t^2-\sigma^2)e^{-k^2t^2/2\sigma^2}\nonumber\\
&\geq 0\Leftarrow t\geq \dfrac{\sigma}{k}\quad,\forall k\geq1\;.\nonumber\\
&\text{Therefore, $f(t)$ is convex if } t\geq \sigma\;.\label{LowerBound_ConvexityOfEs}\\
\intertext{Now using Lemma \ref{Lemma:PoissonSum}, we get}
f(t)&=\sqrt{2\pi}\>\dfrac{\sigma}{t}\Big(1+\sum\limits_{k=0}^{\infty}e^{-2\pi^2k^2 \sigma^2/t^2}\Big)\nonumber\\
\Rightarrow f''(t)&=2\sqrt{2\pi} \>\dfrac{\sigma}{t^3}\Big(1+\dfrac{1}{x^4}\big(x^4-10\pi^2k^2\sigma^2x^2\nonumber\\&\qquad\qquad\qquad\qquad\qquad\quad+8\pi^4k^4\sigma^4\big)e^{-2\pi^2k^2 \sigma^2/t^2}\Big)\nonumber\\
&\geq 0\Leftarrow t\notin \Big(\pi\sigma\sqrt{5-\sqrt{17}}k,\,\pi\sigma\sqrt{5+\sqrt{17}}k\Big),\,\forall k\geq1\;.\nonumber\\
&\text{Therefore, $f(t)$ is convex if } t\leq 2.9\,\sigma\;.\label{UpperBound_ConvexityOfEs}
\end{alignat}
The proof is complete by combining the conditional convexity of $f(t)$ given in \eqref{LowerBound_ConvexityOfEs} and \eqref{UpperBound_ConvexityOfEs}.
\end{proof}

\begin{lemma}\label{Lemma:ConvexityOfEw}
The function $g(t)=tf(t)$ is convex in $t$\;.
\end{lemma}
\begin{proof}We prove the convexity in a way similar to Lemma \ref{Lemma:ConvexityOfEs}.
\begin{alignat}{2}
f(t)&=t\Big(1+\sum\limits_{k=1}^{\infty}e^{-k^2t^2/2\sigma^2}\Big)\nonumber\\
\Rightarrow f''(t)&=\sum\limits_{k=1}^{\infty}\dfrac{k^2t}{\sigma^4}(k^2 t^2-3 \sigma^2)e^{-k^2t^2/2\sigma^2}\nonumber\\
&\geq 0\Leftarrow t\geq \sqrt{3}\,\dfrac{\sigma}{k}\quad,\forall k\geq1\;.\nonumber\\
&\text{Therefore, $f(t)$ is convex if } t\geq 1.74\,\sigma\;.\label{LowerBound_ConvexityOfEw}\\
\intertext{Now using Lemma \ref{Lemma:PoissonSum}, we get}
f(t)&=\sqrt{2\pi}\,\sigma\Big(1+\sum\limits_{k=0}^{\infty}e^{-2\pi^2k^2 \sigma^2/t^2}\Big)\nonumber\\
\Rightarrow f''(t)&=\dfrac{4\sqrt{2\pi^5}\, k^2\sigma^3}{t^6}\big(4\pi^2k^2 \sigma^2  - 3 t^2\big)e^{-2\pi^2k^2 \sigma^2/t^2}\nonumber\\
&\geq 0\Leftarrow t\leq \dfrac{2\pi\sigma}{\sqrt{3}}\cdot k\quad,\forall k\geq1\;.\nonumber\\
&\text{Therefore, $f(t)$ is convex if } t\leq 3.6\, \sigma\;.\label{UpperBound_ConvexityOfEw}
\end{alignat}
The proof is complete by combining the conditional convexity of $f(t)$ given in \eqref{LowerBound_ConvexityOfEw} and \eqref{UpperBound_ConvexityOfEw}.
\end{proof}

\begin{prop}\label{prop:convexity_rayleigh}
The energy penalty $\mathcal{E}(\ts )$ is convex in $\ts\;$.
\end{prop}
\begin{proof}
The proof is straightforward from Lemma \ref{Lemma:ConvexityOfEs} and \ref{Lemma:ConvexityOfEw}.
\end{proof}\par
With the convexity in hand, we can use known optimisation techniques\cite{boyd} to solve $\hat{\mathcal{P}}$. Later in Section \ref{sec:numerical}, we will use a simple bisection algorithm for this optimisation.

\section{Solution to the Problem $\mathcal{P}$}\label{sec:analysis2}
In this section, we consider the optimisation problem $\mathcal{P}$ and find the optimal sampling interval and offset $\{\ts^*,\delta^*\}$. 
Before going into the details, it is worth mentioning the motivation behind the introduction of the offset $\delta$ to the model. 
In practice, there is a minimum time threshold before which the event does not occur. For example, in a WCA, a human user takes a strictly non-zero minimum amount of time to finish a task and sampling before this minimum threshold simply adds to the energy wastage. The addition of an offset to the policy mitigates this energy wastage by not allowing the first sample to fall before a given time. However, this threshold is not necessary for benefiting from the offset. Even in the absence of a threshold, an offset in sampling is particularly useful in situations where the TTE distribution has a small variance and/or large mode. Here, the probability of the event occurring before the first 
sample is small enough so that the expected energy saving in delaying the sampling outweighs the expected cost of the encountered wait due to a (potential) event. 
As a result, the optimum value of an offset is always greater than or equal to the minimum threshold.

\subsection{General TTE Distribution}
As in the earlier section, we start by assuming a general TTE distribution and later extend it to particular distributions.
\subsubsection{Expected number of samples}
The probability that the number of samples taking any given integer value can be computed from the CDF and CCDF of of the TTE.
\begin{alignat}{2}
    \mathbb{P}(\mathcal{S}=1)&=F_{\mathcal{T}}\big(\delta\big)\nonumber\\
    \mathbb{P}(\mathcal{S}=k)&=F_{\mathcal{T}}\big((k-1)\ts +\delta\big)-F_{\mathcal{T}}\big((k-2)\ts +\delta\big),\ \forall\ k\geq2\nonumber\\
    \Rightarrow\mathbb{P}(\mathcal{S}\leq k)&=F_{\mathcal{T}}\big((k-1)\ts +\delta\big),\ \forall\ k\geq1\nonumber\\
    \Rightarrow\mathbb{P}(\mathcal{S}>k)&=\bar{F}_{\mathcal{T}}\big((k-1)\ts +\delta\big),\ \forall\ k\geq1\;.\label{eq:cpmf_general_offset}
\end{alignat}
The expected number of samples can be computed as
\begin{alignat}{2}
 \mathbb{E}[\mathcal{S}]&=\sum_{k=0}^{\infty}\mathbb{P}(\mathcal{S}>k)\nonumber\\
 &=1+\sum_{k=1}^{\infty}\bar{F}_{\mathcal{T}}\big((k-1)\ts +\delta\big)\tag{from \eqref{eq:cpmf_general_offset}}\nonumber\\
 \Rightarrow \mathbb{E}[\mathcal{S}]&=1+\sum_{k=0}^{\infty}\bar{F}_{\mathcal{T}}\big(k\ts +\delta\big)\;.\label{Es_general_offset}
\end{alignat}

\subsubsection{Expected wait time}
Recall that the offset $\delta\geq\ts$ and that the maximum possible wait $w$ is upper-bounded by the sampling interval (at the time of the event). As a result, a value of $w$ such that $\ts <w\leq \delta$ can arise only before the first sample while a value $w\leq \ts $ can potentially arise before any sampling instance. The CDF of wait penalty $F_\mathcal{W}(w)$ can be obtained by taking the probability of the TTE to fall at most $w$ short of any sampling instance. 
\begin{alignat}{2}
F_\mathcal{W}(w)&=\begin{cases}
     \sum\limits_{k=0}^{\infty}\big(F_{\mathcal{T}}(k\ts +\delta)-F_{\mathcal{T}}(k\ts +\delta-w)\big) & \text{if }w\leq \ts \\
      1-F_{\mathcal{T}}\big(\delta-w\big) & \text{if }\ts <w\leq \delta\\
      1 & \text{if }w>\delta
    \end{cases}\nonumber  \\
\intertext{Since $\mathcal{W}$ is a non-negative random variable, we have}
\mathbb{E}[\mathcal{W}]&={\int_{0}^{\infty}}\big(1-F_\mathcal{W}(w)\big)\,\dif w\nonumber\\
&=\int_{0}^{\ts }\Big(1-\sum_{k=0}^{\infty}\big(F_{\mathcal{T}}(k\ts +\delta)-F_{\mathcal{T}}(k\ts +\delta-w)\big)\Big)\,\dif w\nonumber\\&\qquad\qquad\qquad\qquad\qquad\quad +\int_{\ts }^{\delta}F_{\mathcal{T}}\big(\delta-w\big)\,\dif w\;.\nonumber
\intertext{Since integral limits and infinite sum are finite, and the summand is non-negative, we get}
\mathbb{E}[\mathcal{W}]&=\ts -\sum_{k=0}^{\infty}\int_{0}^{\ts }\big(F_{\mathcal{T}}(k\ts +\delta)-F_{\mathcal{T}}(k\ts +\delta-w)\big)\,\dif w\nonumber\\&\qquad\qquad\qquad\qquad\qquad\quad+\int_{\ts }^{\delta}F_{\mathcal{T}}\big(\delta-w\big)\,\dif w\;.\label{Ew_general_offset}
\end{alignat}
\par As expected, a substitution of $\delta=\ts $ in \eqref{Es_general_offset} or \eqref{Ew_general_offset} under the policy $\Pi$ returns the same result that we obtained under the policy $\hat{\Pi}$ discussed in Section \ref{sec:analysis}.

\subsection{Exponentially Distributed TTE}
Due to the memory-less property of the distribution, it is predictable that the optimum offset under an exponentially distributed TTE takes the default value of $\delta=\ts $; thus making the solution of $\mathcal{P}$ and $\hat{\mathcal{P}}$ one and the same. Before formally proving this in Proposition \ref{prop:OffsetIsIrrelevant}, we will first compute the energy penalty $\mathcal{E}(\ts,\delta)$ by substituting the CCDF and CDF of the exponential distribution in \eqref{Es_general_offset} and \eqref{Ew_general_offset}, respectively. We get
\begin{alignat}{2}
    &\qquad\;\;\;\mathbb{E}[\mathcal{S}]=1+\dfrac{e^{-\lambda \delta}}{1-e^{-\lambda \ts }}\nonumber\\
    &\qquad\;\;\mathbb{E}[\mathcal{W}]=\dfrac{(1-\lambda \delta)e^{-\lambda \ts }+\lambda \ts \,e^{-\lambda \delta}+\lambda \delta -1}{\lambda(1-e^{-\lambda \ts })}\nonumber\\
&\Rightarrow\mathcal{E}(\ts ,\delta)=\alpha\mathbb{E}[\mathcal{S}]+\beta\mathbb{E}[\mathcal{S}]\nonumber\\
        &=\dfrac{e^{-\lambda \ts }\big(\beta\!-\!\lambda(\beta\delta\!+\!\alpha)\big)+\lambda\big(\beta\delta\!+\!\alpha\!+\!e^{-\lambda\delta}(\beta \ts \!+\!\alpha)\big)-\beta}{\lambda(1-e^{-\lambda \ts })}\;.\label{eq:energypenalty_exponential_offset}
\end{alignat}
\begin{prop}\label{prop:OffsetIsIrrelevant}
Given that the TTEs are exponentially distributed,  $\delta^* = \ts^*$ under an optimal policy for $\mathcal{P}$; therefore, $\mathcal{P}$ and $\hat{\mathcal{P}}$ are equivalent.
\end{prop}
\begin{proof}
Let $\pi^{(t_1,t_2\dots)}$ be a sampling policy with sampling intervals given by $t_1,\,t_2,\dots$ and let $\mathcal{E}\big(\pi^{(t_1,t_2\dots)}\big)$ be the associated energy penalty. Let $\ts ^{\#}$ be the optimum sampling interval obtained by solving $\hat{\mathcal{P}}$, where there is no offset. That is
\begin{equation}
    \mathcal{E}\big(\pi^{(t,t,t,\dots)}\big)\geq \mathcal{E}\big(\pi^{(\ts ^{\#},\ts ^{\#},\ts ^{\#},\dots)}\big),\;\forall t\geq0\;.\label{eq:eq1_Lemma:OffsetIsIrrelevant}
\end{equation}
Now, let $\ts ^*$ and $\delta^*$ be the solution to $\mathcal{P}$ resulting in an optimal policy $\pi^{(\delta^*,{\ts}^*,{\ts}^*,{\ts}^*,\dots)}$ and let $F_{\mathcal{T}|\delta^*}$ be the conditional CDF of the TTE given $\mathcal{T}>\delta^*$. If the problem to compute the optimum sampling interval and phase is repeated at time $\delta^*$ using the conditional CDF, we get a new solution containing an optimum sampling interval $\tilde{\ts}^*$ and an optimum (second) offset $\tilde{\delta}^*$. 
As this updated set of sampling policy denoted by $\pi^{(\delta^*,\tilde{\delta}^*,\tilde{\ts}^*,\tilde{\ts}^*,\dots)}$ is a superset of the previously obtained sampling policy $\pi^{(\delta^*,{\ts}^*,{\ts}^*,{\ts}^*,\dots)}$, the new solution provides us with an energy penalty that is not worse than the one obtained before. 
However, due to the memory-less property of the exponential distribution we know that $$F_{\mathcal{T}}(t)=F_{\mathcal{T}|\mathcal{T}>k\delta^*}(t+k\delta^*),\;\forall k\in\mathbb{N}\;.$$
Therefore by using the conditional CDF, optimisation problem is essentially unchanged and the newly obtained solution will be numerically equal to the corresponding values obtained in the previous step; i.e., $\tilde{\ts}^*={\ts}^*$ and $\tilde{\delta}^*={\delta}^*$. By repeating this process supported by similar arguments, we get a sequence of non-increasing energy penalties. That is
\begin{alignat}{2}
\mathcal{E}\big(\pi^{(\ts ^{\#},\ts ^{\#},\dots)}\big)&\geq\mathcal{E}\big(\pi^{(\delta^*,\ts ^*,\ts ^*,\dots)}\big)\tag{as     ${\hat{\Pi}}\subseteq{\Pi}$} \nonumber\\
&\geq\mathcal{E}\big(\pi^{(\delta^*,\delta^*,\ts ^*,\ts ^*,\dots)}\big)\nonumber\\
&\vdots\nonumber\\
&\geq\mathcal{E}\big(\pi^{(\delta^*,\delta^*,\dots)}\big)\nonumber\\
&\geq\mathcal{E}\big(\pi^{(\ts ^{\#},\ts ^{\#},\dots)}\big)\tag{from \eqref{eq:eq1_Lemma:OffsetIsIrrelevant}}\\
\Rightarrow \mathcal{E}\big(\pi^{(\ts ^{\#},\ts ^{\#},\dots)}\big)&=\mathcal{E}\big(\pi^{(\delta^*,\delta^*,\dots)}\big)\;.\nonumber
\end{alignat}
Therefore, $\delta^*=\ts ^{\#}$ as a result of the convexity of $\mathcal{E}(\ts )$ from Proposition \ref{prop:convexity_rayleigh}. Thus, the optimum set of sampling instances under ${\Pi}$ is $\{t=k\ts ^*,\,k\in\mathcal{N}^+\}$, which is essentially the solution to $\hat{\mathcal{P}}$, proving that $\mathcal{P}\equiv\hat{\mathcal{P}}$.
\end{proof}

As a result of the equivalence of $\mathcal{P}$ and $\hat{\mathcal{P}}$, we do not benefit from proceeding with the relatively complex optimisation of \eqref{eq:energypenalty_exponential_offset} and we can use the solution from Proposition \ref{prop:exponentialSolution} to find the optimum energy penalty.

\subsection{Rayleigh Distributed TTE}
\begin{lemma}
The energy penalty $\mathcal{E}$ under Rayleigh distributed TTE is given by
\begin{equation}
    \mathcal{E}(\ts ,\delta)=(\alpha+\beta \ts )\sum_{k=0}^{\infty}e^{-(k\ts +\delta)^2/2\sigma^2} - \beta\sigma\sqrt{\tfrac{\pi}{2}}+\beta\delta+\alpha\;.\label{eq:energyPen_rayleigh_offset}
\end{equation}
\end{lemma}
\begin{proof}
Substituting the CDF of the Rayleigh distribution, $F_{\mathcal{T}}(t)=1-e^{-t^2/2\sigma^2}$ in \eqref{Es_general_offset} or \eqref{Ew_general_offset} gives
\begin{alignat}{2}
\mathbb{E}[\mathcal{S}]&=1+\sum_{k=0}^{\infty}e^{-(k\ts +\delta)^2/2\sigma^2}\label{Es_rayleigh_offset}\\
\mathbb{E}[\mathcal{W}]&=\delta+\ts \sum_{k=0}^{\infty}e^{-(k\ts +\delta)^2/2\sigma^2}-\sigma\sqrt{{\dfrac{\pi}{2}}}\;.\label{Ew_rayleigh_offset}
\end{alignat}
The energy penalty $\mathcal{E}(\ts ,\delta)$ can be computed by adding $\mathbb{E}[\mathcal{S}]$ and $\mathbb{E}[\mathcal{W}]$ using the weights $\alpha$ and $\beta$, respectively.
\end{proof}
\begin{algorithm}[t]
\SetAlgoLined
 Initialise the upper limit of optimising variables, $\ts ^{(\text{max})}$ and $n^{(\text{max})}$;\\
 Initialise stopping criterion $\xi\,$;\\
 $n^{(\text{L})}\gets1\,;\;n^{(\text{H})}\gets n^{(\text{max})}\,$;\\
\For{$i\leftarrow 0$ \KwTo $\ceil{\log _2(n^{(\text{max})}-1)}$}{
  $n\gets (n^{(\text{L})}+n^{(\text{H})})/2\,$;\\
  $\ts ^{(\text{L})}=0;\;\ts ^{(\text{H})}=\ts ^{(\text{max})}\,$;\\
  \While{$\ts ^{(\mathrm{H})}-\ts ^{(\mathrm{L})}>\xi$}{
    $\ts \gets(\ts ^{(\text{L})}+\ts ^{(\text{H})})/2\,$;\\
    \eIf{$\tfrac{\partial}{\partial \ts }\mathcal{E}(\ts ,n\ts )\geq0$}{
   $\ts ^{(\text{H})}\gets \ts \,$;
   }{
   $\ts ^{(\text{L})}\gets \ts \,$;
  }
  }
\eIf{$\mathcal{E}(\ts ,n\ts )-\mathcal{E}(\ts ,(n-1)\ts )\geq0$}{
   $n^{(\text{H})}\gets n\,$;
   }{
   $n^{(\text{L})}\gets n\,$;
  }
 }
 $\ts ^*=\ts \,;\;\delta^*=n\ts ^*\,;$
 \caption{Algorithm to find optimum sampling interval $\ts ^*$ and optimum offset $\delta^*$ with the constraint $\delta=n\ts ,\, n\in\mathbb{N}^+$.}\label{Algo}
\end{algorithm}
Proving the convexity of $\mathcal{E}(\ts ,\delta)$ in the two variables $\ts $ and $\delta$ is hard. The difficulty arises as we need to prove that the Hessian is non-negative and showing this for the infinite sum is highly non-trivial. Furthermore, the proof that we used in Lemma \ref{Lemma:PoissonSum} cannot be extended as the function is not even in $\ts $ after the addition of the second variable $\delta$. 
Instead, we will consider a sub-optimal solution by assuming that the offset is an integer multiple of the sampling interval. We will now show that the energy penalty is convex in $\ts $ for a fixed offset of the form $\delta=n\ts$.

\begin{lemma}\label{Lemma:DiscreteConvexity}
The function $f_n(t)=(a+bt)\sum_{k=0}^{\infty}e^{-((k+n)t)^2/2\sigma^2}$ is convex in $t$ for $n\in \mathbb{N}\;$.
\end{lemma}
\begin{proof}
See that $f_n(t)$ for any $n$ is obtained by removing first $n-1$ terms from $f_0(t)$. From Proposition \ref{prop:convexity_rayleigh}, we know that $f_0(t)$ is convex. 
\begin{equation*}
    \Rightarrow 0\leq f_0''(t)=\sum_{k=0}^{\infty}\dfrac{k^2t^2(a+bt)-s^2(a+3bt)}{\sigma^2} k^2e^{{-(k+n)^2t^2}/{2\sigma^2}}\;.
\end{equation*}
As the terms in $f_0''(t)$ are non-decreasing with $k$, removing first $n-1$ terms will not affect the sign of the sum and therefore $f_n(t)$ is also convex.
\end{proof}

We can see with the help of \eqref{eq:jacobi} and \eqref{eq:energyPen_rayleigh_offset} that the energy penalty is the sum of a linear function and a continuously decreasing and convex Jacobi theta function. This shows the presence of a single minima with respect to $\delta$ for any fixed $\ts $. 
With this information and the convexity in $\ts $ for any fixed $\delta$ (Lemma \ref{Lemma:DiscreteConvexity}), we use Algorithm \ref{Algo} to find the sampling interval and offset that attains a minimal energy penalty while satisfying the constraint $\delta=n\ts ,\;n\in\mathbb{N}$. In the next section, we will compare this energy penalty to $\mathcal{E}^*$ obtained using a brute force search and show that Algorithm \ref{Algo} attains near-optimality. 
As a direct result of this near-optimality, we will reuse the same notations $(\cdot)^*$ to represent the solution obtained using Algorithm \ref{Algo}, unless otherwise necessary.

\section{Numerical Results}\label{sec:numerical}
In this section, we present the numerical results to illustrate the behaviour of the energy penalty, the gain achieved by the proposed design, and their dependency on the system parameters. 
We consider two parameter settings that are motivated by the characteristics of the CPS and the VAS. We start with describing the characterisation of both of these systems and the underlying distribution of the observed process. Afterwards, we compare the proposed design(s) with a baseline sampling policy that is used in practice. Furthermore, to verify the near-optimality of Algorithm \ref{Algo} in solving $\mathcal{P}$, we compare its solution with that of a brute force method. For the CPS with an exponentially distributed TTE, the offset is irrelevant and hence we are not using Algorithm \ref{Algo}. As a result, this comparison with a brute-force solution is required only for the VAS but not for the CPS.
\subsection{Parameter Settings}\label{prelim}
\begin{figure*}[t]
\centering
\begin{subfigure}{0.32\linewidth}
\centerline{\includegraphics[width=\linewidth]{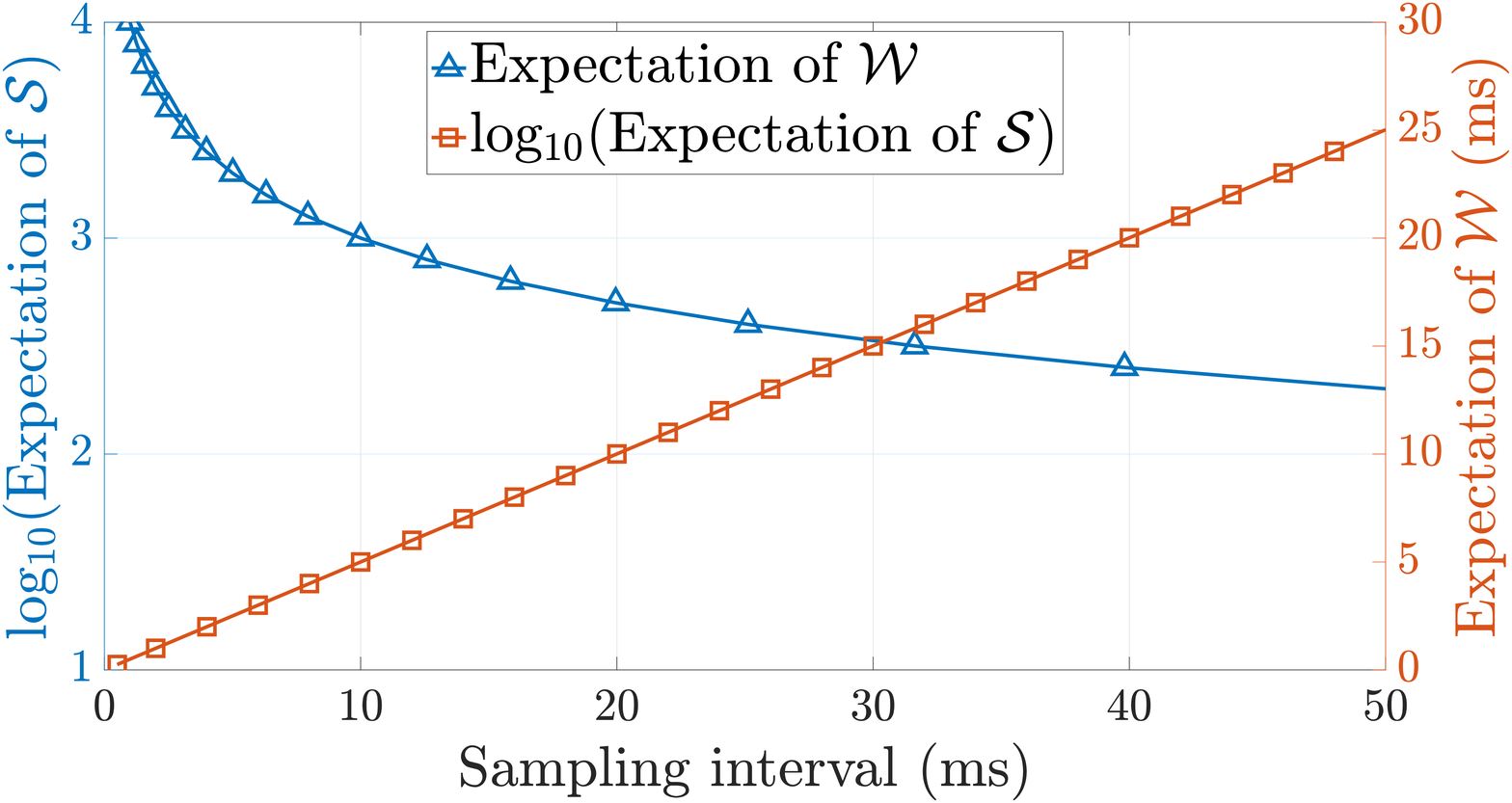}}
    \caption{$\mathrm{log}_{10}(\mathbb{E}[\mathcal{S}]),\,\mathbb{E}[\mathcal{W}]$ vs. $T_s$.}
\label{fig:e1}
\end{subfigure}%
~
\begin{subfigure}{0.32\linewidth}
\centerline{\includegraphics[width=\linewidth]{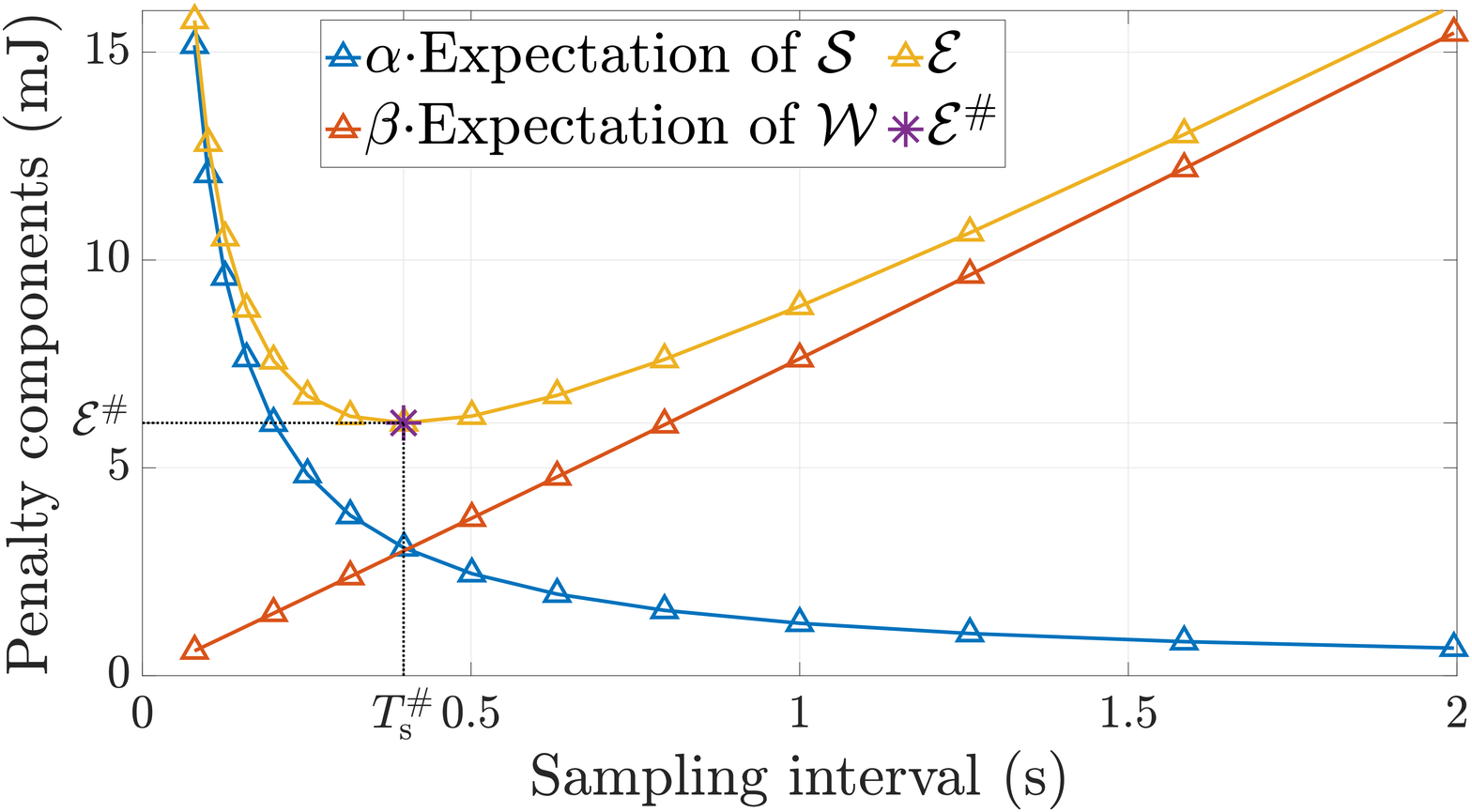}}
\caption{Penalty components, $\mathcal{E}$ vs. $\ts$.}
\label{fig:e2}
\end{subfigure}%
~
\begin{subfigure}{0.32\linewidth}
\centerline{\includegraphics[width=\linewidth]{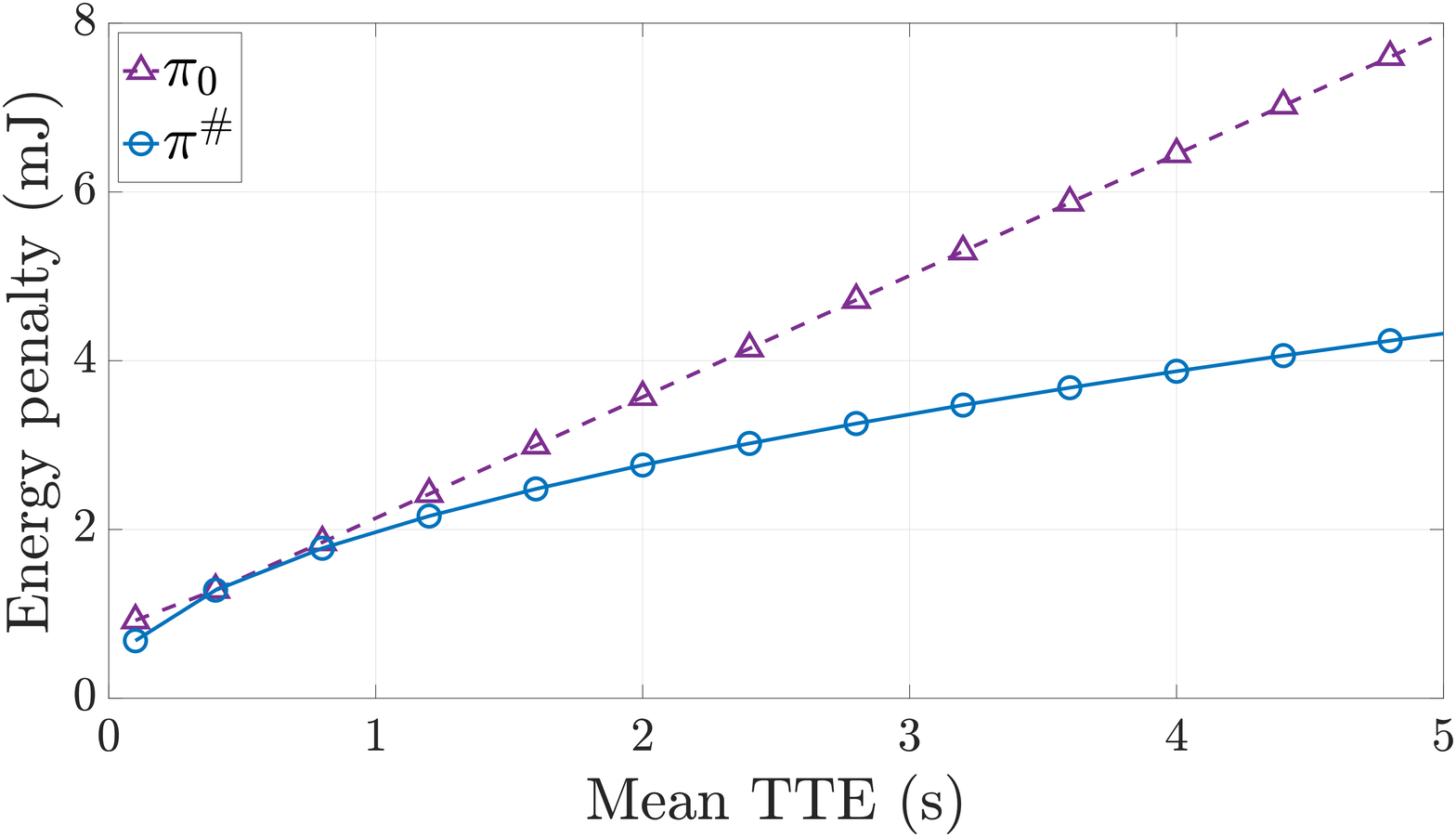}}
\caption{$\mathcal{E}$ vs. $\mathbb{E}[\mathcal{T}]$.}
\label{fig:e3}
\end{subfigure}
\caption{Expected number of samples $\mathbb{E}[\mathcal{S}]$, expected wait $\mathbb{E}[\mathcal{W}]$, energy penalty $\mathcal{E}$ and its components plotted against the sampling interval $\ts$ and mean TTE $\mathbb{E}[\mathcal{T}]$ for the CPS. Proposed policy $\pi^\#$ is compared with the baseline policy $\pi_0$.}
\label{fig:exp}
\end{figure*}
\begin{figure*}[ht]
\centering
\begin{subfigure}{0.32\linewidth}
\centerline{\includegraphics[width=\linewidth]{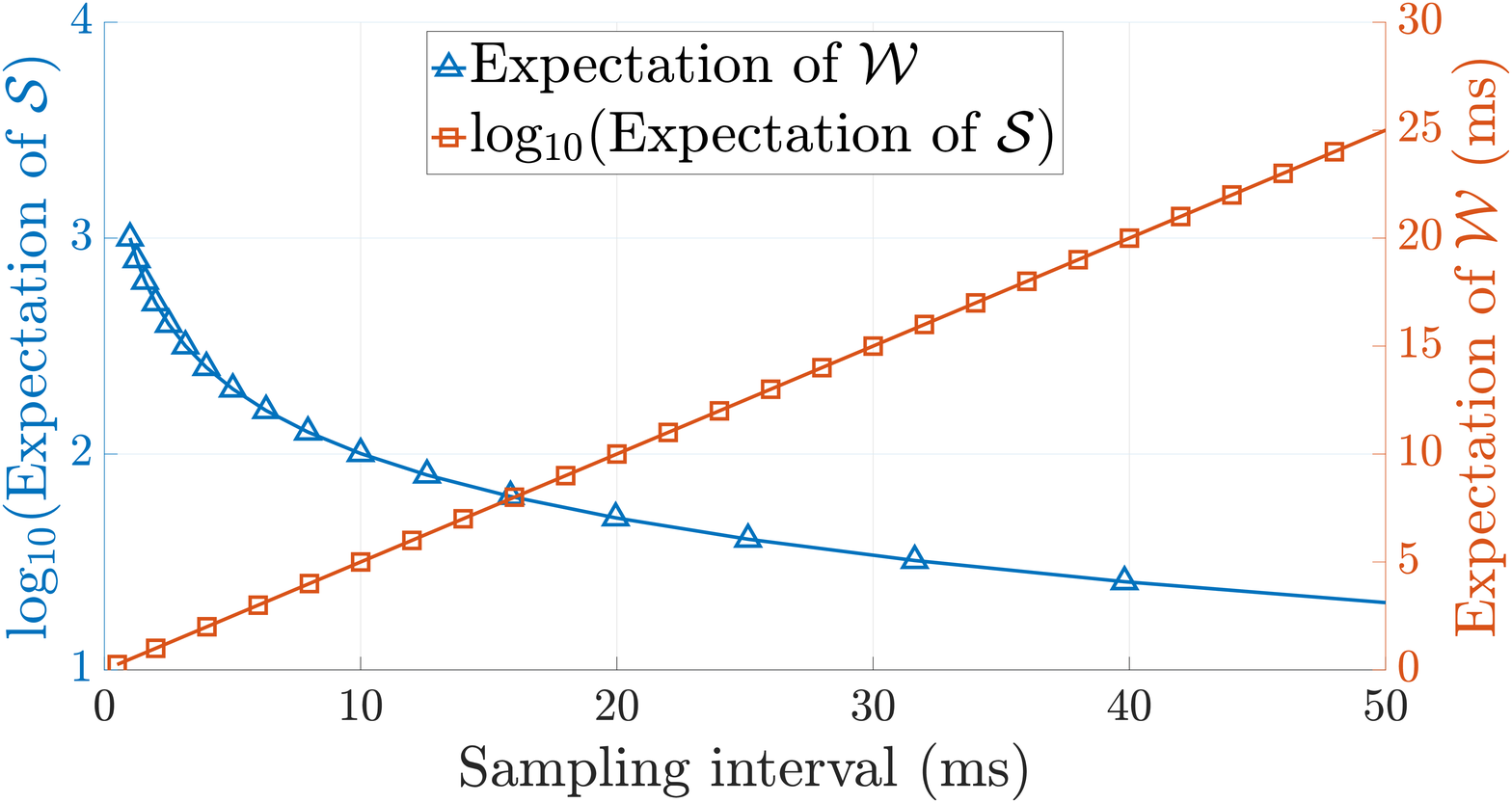}}
\caption{$\mathrm{log}_{10}(\mathbb{E}[\mathcal{S}]),\,\mathbb{E}[\mathcal{W}]$ vs. $T_s$.}
\label{fig:r1}
\end{subfigure}%
~
\begin{subfigure}{0.32\linewidth}
\centerline{\includegraphics[width=\linewidth]{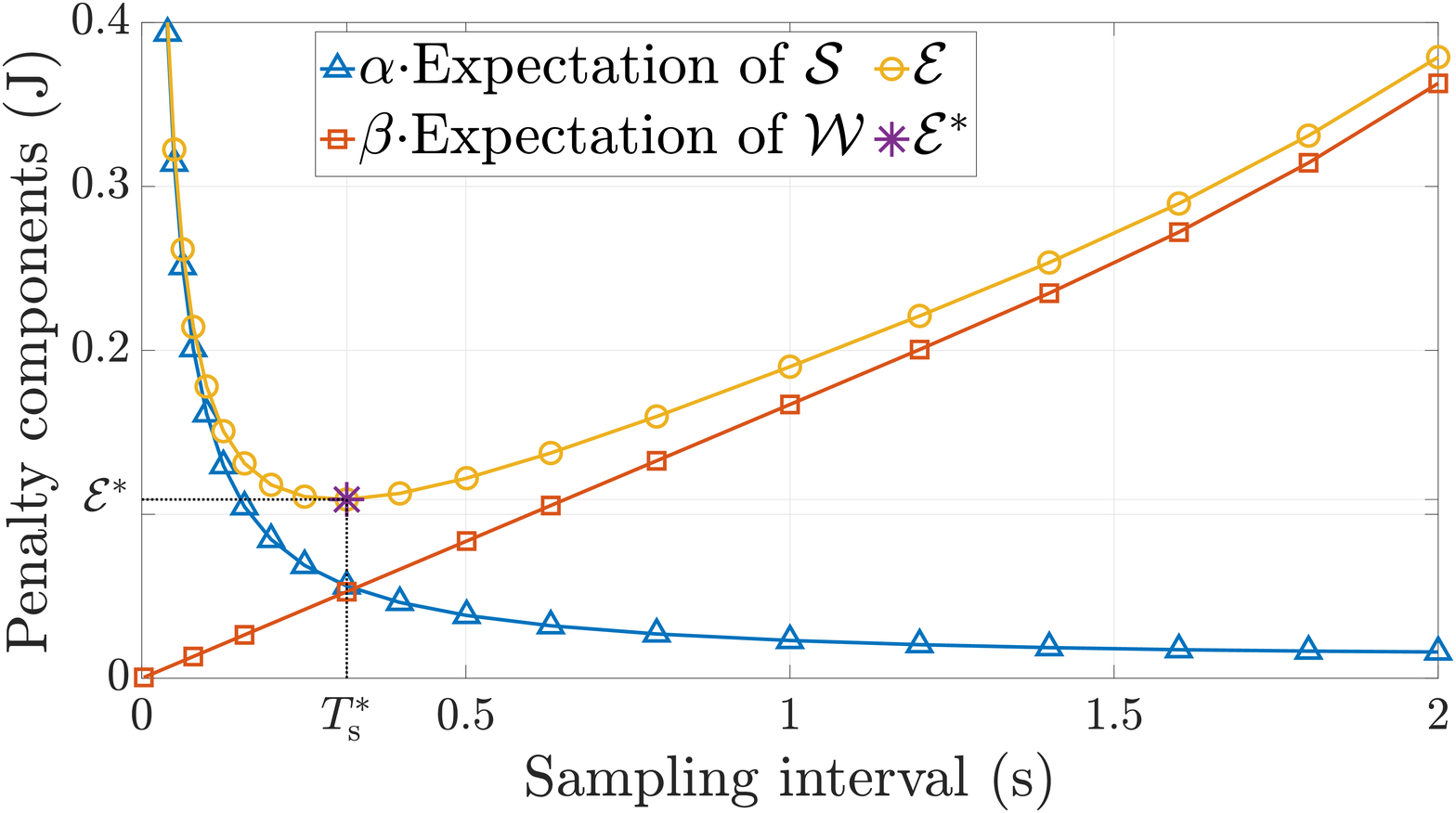}}
\caption{Penalty components, $\mathcal{E}$ vs. $\ts$.}
\label{fig:r2}
\end{subfigure}%
~
\begin{subfigure}{0.32\linewidth}
\centerline{\includegraphics[width=\linewidth]{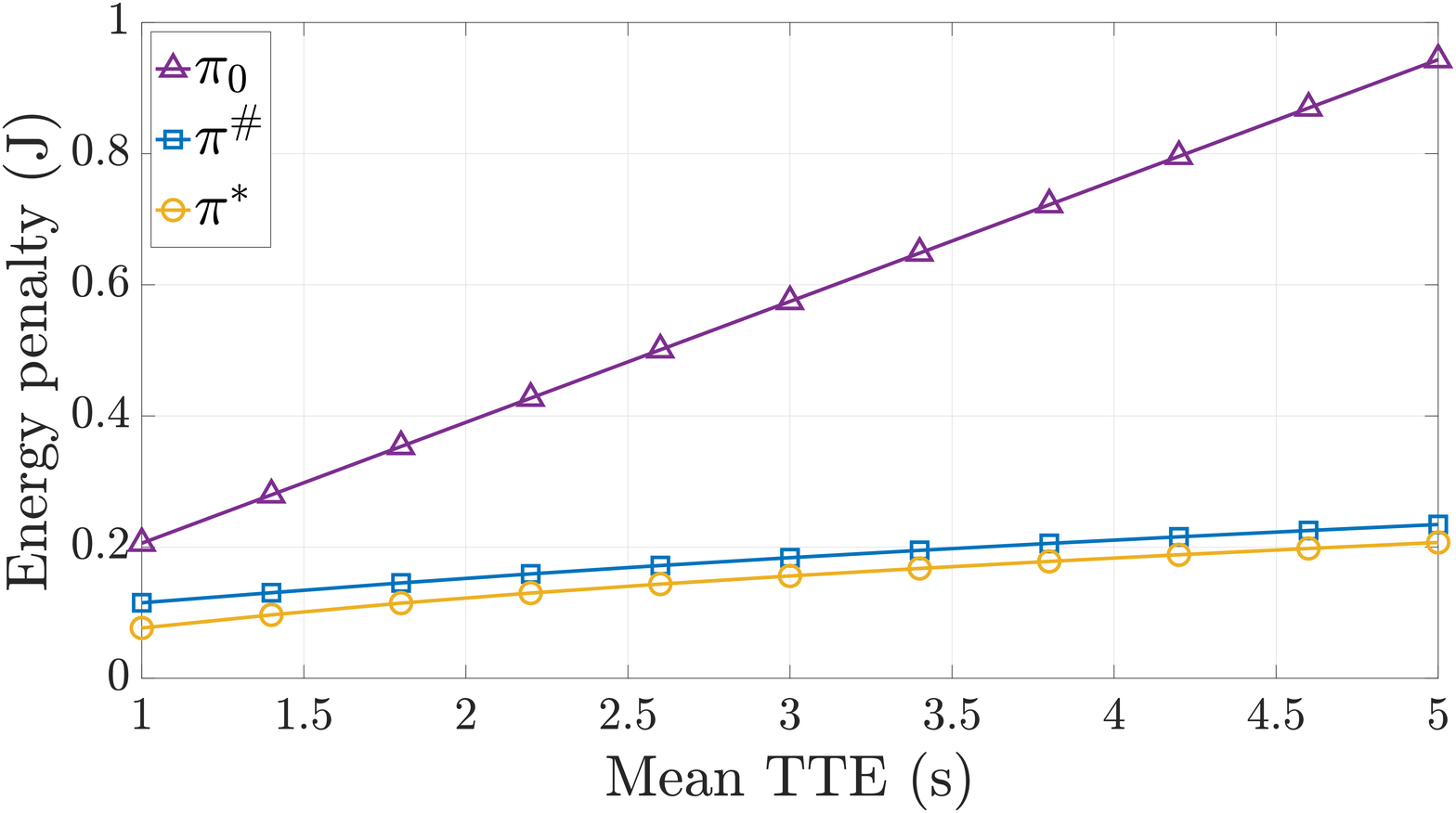}}
\caption{$\mathcal{E}$ vs. $\mathbb{E}[\mathcal{T}]$.}
\label{fig:r3}
\end{subfigure}
\caption{Expected number of samples $\mathbb{E}[\mathcal{S}]$, expected wait $\mathbb{E}[\mathcal{W}]$, energy penalty $\mathcal{E}$ and its components vs. the sampling interval $\ts$ and mean TTE $\mathbb{E}[\mathcal{T}]$ for the VAS. Proposed policies $\pi^\#$ and $\pi^*$ are compared with the baseline policy $\pi_0$.}
\label{fig:rayleigh}
\end{figure*}
For the CPS, we consider a failure detection system where the TTE is exponentially distributed with a mean of 10s. To detect the events (i.e., failures for a CPS), the terminal samples the process and sends the sample to the back-end via a low-power, low-throughput transmitter. For the device, we refer to the wireless sensor network characterisation in \cite{PowConsumptionAssessmentinWSN} with the data size fixed at 127 Bytes. For communication technology, we choose 802.15.4 operating in the Sub1GHz ISM band with an approximate throughput of 250kbps, resulting in a communication delay of 4ms. The device is assumed to work at 3V drawing a current of 15mA during communication and 5mA during idling. We also assume that the processing time of the successful sample $\tau_{\text{s}}=5$ms. 

On the other hand, a VAS typically need to transmit video frames, thus requiring a larger throughput. We refer to the WCA Lego experiment \cite{TowardsWCA} where a set of tasks performed by a human user are monitored to detect the task completion by taking snapshots of the progress. The frame size used in the experiment is roughly 300kB ($640\times 480$ resolution) and the mean task time observed is 4.846s. The processing time observed at the back end is approximately 525ms. For the terminal device, we consider a Google Glass using an 802.11ax transmitter providing a data rate of 400Mbps thus resulting in a 5.85ms communication delay for each snapshot. The Google Glass typically takes 334mW during \textit{active/screen-off} and 2960mW during \textit{video chat}\cite{googleGlass}. Thus, we assume these power figures as the idle power and the communication power, respectively. On top of this, we also assume a minimum possible task time $t_{\text{min}}$ (response time of the human user) of 0.5s. The parameters used for the CPS and the VAS are presented in TABLE \ref{tab:Example}.
\begin{table}[htb]
    \centering
    \resizebox{\columnwidth}{!}{
    \begin{tabular}{|c||c|c|c|c|c|c|c|}
    \hline
         &$P_0$&$P_\text{c}$&$\tau_\text{c}$&$\tau_{\text{s}}$&$\nicefrac{\beta}{\alpha}$&$\mathbb{E}[\mathcal{T}]$&$t_{\text{min}}$  \\ \hline\hline
         CPS&15mW&45mW&4ms&5ms&125&10s&--\\ \hline
         VAS&334mW&2.96W&5.85ms&525ms&21.7&4.84s&0.5s\\ \hline
    \end{tabular}
    }
    \caption{Parameters used for the cyber-physical system(CPS) and the video analytics system (VAS).}
    \label{tab:Example}
\end{table}

Even though our solution does not depend on the type of energy source, we assume that the CPS and the VAS are powered by a fixed battery pack of arbitrary capacity. As battery life improvement is one of the directly observable results of energy saving, we evaluate the performance of the proposed solution in terms of percentage increase in battery life. Note that as we use a percentage increase, the exact capacity of the battery is irrelevant.
\begin{table}[htb]
\centering
       \begin{tabular}{|l|l||l|l|l|l|}
            \hline
            &&$\pi_0$&${\pi^\#}$&$\pi^*$&$\pi_\text{b}$\\ \hline \hline        
            CPS& $\ts$&$83.3$ ms&$\ts^\#$&-&-\\\hline
 \multirow{2}{*}{VAS} & $\ts$ &$83.3$ ms&$\ts^\#$&$\ts^*$&Brute\\\cline{2-5}
            & $\delta$ & -&-&$\delta^*$&Force\\
              \hline
        \end{tabular}
    \caption{Various sampling policies considered for comparison and their corresponding sampling interval and offset.}
    \label{tab:methods}
\end{table}

For the performance evaluation of the proposed policies $\pi^\#$ and $\pi^*$, we consider a baseline periodic sampling policy denoted by $\pi_0$ as shown in TABLE \ref{tab:methods}. 
For selecting the baseline sampling rate for the VAS, we take a hint from the WCA system \cite{TowardsWCA} that motivated the system model. The authors use a video camera to capture the monitored process that runs at a frame rate of 24fps which ideally results in a sampling interval of $41.6$ms. However, whenever the processing of the discarded samples takes more than $41.6$ms, the next sample is delayed accordingly. As a result, the mean sampling interval is observed to be around $83.3$ms thus justifying a baseline with sampling interval anywhere in between $41.6$ms and $83.3$ms.
That being said, we will see in the next subsection that the optimum sampling interval for the considered VAS is around $300$ms.
As a result, even though our model does not take samples by considering the processing time, we take our baseline sampling to be $83.3$ms to show the minimum performance improvement when the baseline is closest to the optimum
We use the same baseline sampling interval for the CPS for simplicity in comparison. 

To verify the near-optimality of $\pi^*$ obtained using Algorithm \ref{Algo}, we compare it with that of a brute force solution of $\mathcal{P}$. To find this brute force solution, we divide the practically feasible domain of the two-dimensional plane formed by $\ts$ and $\delta$ into grids of size $10^{-3}\times10^{-3}\;\mathrm{s}^2$ and search for the optimum solution (also known as grid search).

As the explanations to some of the figures of the CPS and the VAS involve similar arguments, we explain them concurrently whenever it is convenient.
Also, note that the y-axes of some of these figures are kept different intentionally for better visualisation.
The parameters for all the plots are taken from TABLE \ref{tab:Example} and TABLE \ref{tab:methods}, with a potential exception to the variable parameter under discussion in a particular plot. 

\subsection{Energy Penalty}
In Fig. \ref{fig:e1} and Fig. \ref{fig:r1} we present the the expected number of samples $\mathbb{E}[\mathcal{S}]$ and the expected wait $\mathbb{E}[\mathcal{W}]$ by varying $\ts$ for the CPS and the VAS, respectively. 
For comparison, at $\ts=10$ms, the CPS and the VAS expect 1000 and 100 samples per event, respectively. The expected wait however is approximately the same at around $4$ms.
As noted in Sections \ref{sec:analysis} and \ref{sec:analysis2}, $\mathbb{E}[\mathcal{S}]$ and $\mathbb{E}[\mathcal{W}]$ show opposing behaviour with an increase in $\ts$. 
After a very rapid decrease, $\mathbb{E}[\mathcal{S}]$ goes asymptotically to $1$ with increasing $\ts$. This points to the single sample that is ideally required for the event detection.
On the other hand, $\mathbb{E}[\mathcal{W}]$ shows a gradual (approximately linear) increase with $\ts$. For instance, doubling the sampling interval doubles the wait. 
For a VAS operating at a sampling frequency of 0.5s, this corresponds to approximately $50\%$ increase in the TTF. Recall that TTF is the effective delay experienced by the human user. 

In Fig. \ref{fig:e2} and Fig. \ref{fig:r2}, we show the energy penalty and its components as a function of $\ts$. 
Observe $\mathbb{E}[\mathcal{S}]$ and $\mathbb{E}[\mathcal{W}]$ exhibit an opposing behaviour with a change in $\ts$ resulting in an energy penalty minima in their weighted sum. This minima and the corresponding $\ts$ is also marked in the figures. The qualitative behaviour of the energy penalty and its components is similar for the CPS and the VAS with a minimum value of $6$mJ and $100$mJ, respectively. Another important observation is the change in $\mathcal{E}$ as a result of a deviation of $\ts$ from the optimum. As $\mathcal{E}$ increases much more rapidly with a negative deviation than with a positive deviation, it is particularly important to avoid any negative errors in calculating this optimum. That is, an oversampling can cost much more than an undersampling by a similar amount. These errors can arise due to an incorrect estimation of the TTE statistics or insufficient convergence of the optimisation algorithm.

\subsection{Energy Saving}
In Fig. \ref{fig:e3} and Fig. \ref{fig:r3} we show energy penalty incurred for the CPS and the VAS as a function of the mean TTE for various sampling policies from TABLE \ref{tab:methods}.
Since the mean TTE value is application-specific, plotting energy penalty across mean TTE values provides an insight into how the behaviour of the particular CPS and VAS scales across application scenarios with different parameterisation.
The energy penalty reduction achieved by the proposed policies is clear from these figures. Note that when the sampling interval of $83.3$ms is smaller than $\ts^\#$ or $\ts^*$, the penalty is dominated by the discarded samples. As a result, the difference in $\mathcal{E}$ increases with an increase in mean TTE. This increase is much more prominent for the VAS.

The increase in battery life is a more direct measure of energy saving from an application perspective. In Fig. \ref{fig:batterlife} we present this by plotting the percentage increase in the expected battery life of the terminal. In Fig. \ref{fig:batterlife_baseline1} we choose the baseline $\pi_0$ whereas in Fig. \ref{fig:batterlife_baseline2} we choose a different baseline sampling interval of $2$s. We can see that irrespective of the baseline or the application considered, there is an observable increase in battery life. Particularly for the VAS, the proposed sampling policy improves the battery life by $36\%$ over $\pi_0$ at the considered mean TTE of 4.84s.

The motivation to consider two baselines is the value of their sampling interval relative to the optimum sampling interval (for the TTE range considered). 
The sampling interval of $0.83$ms is much smaller, and the sampling interval 2s is much larger than the optimal sampling intervals. Therefore, the former results in oversampling and the latter results in undersampling which results in shifting the dominating component of the energy penalty from  $\mathbb{E}[\mathcal{S}]$ to $\mathbb{E}[\mathcal{W}]$. As a result, the increase in battery life is a increasing function of the mean TTE in Fig. \ref{fig:batterlife_baseline1}, whereas it is a decreasing function in Fig. \ref{fig:batterlife_baseline2}. 
Note that the optimum sampling interval is positively correlated with the mean task time and the battery life is inversely proportional to the energy. 
Therefore, this change of percentage increase in battery life is in alignment with the variation of $\mathbb{E}[\mathcal{S}]$ and $\mathbb{E}[\mathcal{W}]$ with $\ts$ observed earlier in Fig. \ref{fig:exp} and Fig. \ref{fig:rayleigh}.
\begin{figure}[htb]
\centering
\begin{subfigure}{0.8\linewidth}
\centering
\includegraphics[width=\linewidth]{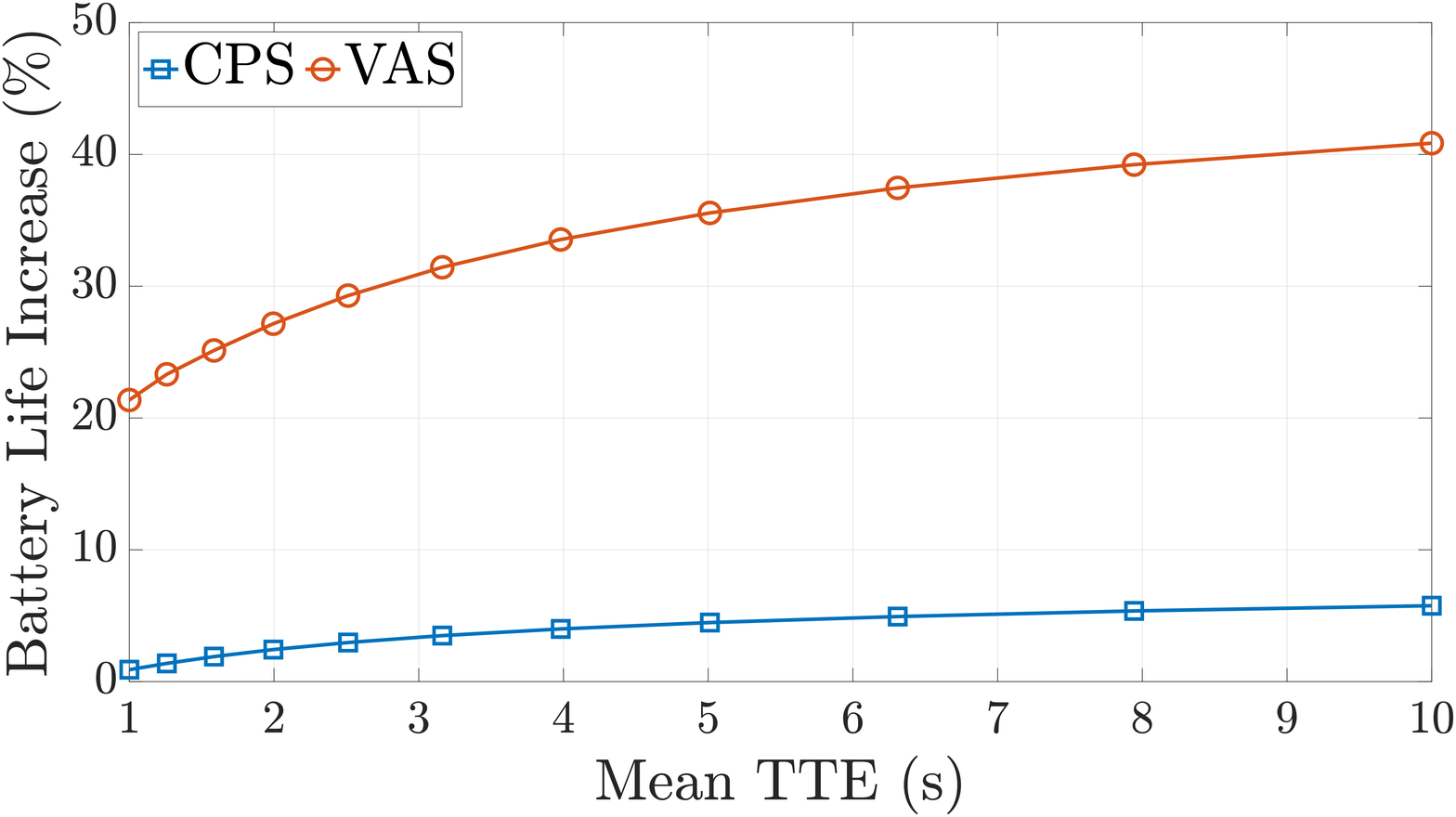}
\caption{Baseline policy $\pi_0$}
\label{fig:batterlife_baseline1}
\end{subfigure}
\begin{subfigure}{0.8\linewidth}
\centering
\includegraphics[width=\linewidth]{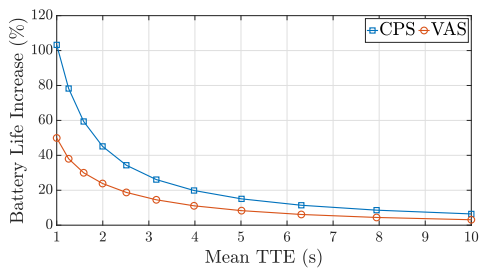}
\caption{Baseline policy with $\ts=2$s.}
\label{fig:batterlife_baseline2}
\end{subfigure}%
\caption{Percentage increase in battery life achieved by the policy $\pi^*$ vs. mean TTE $\mathbb{E}[\mathcal{T}]$ for the CPS and the VAS. The increase is calculated by considering two baselines policies: (a) $\pi_0$, and (b) a policy with sampling interval $\ts=2$s.}
\label{fig:batterlife}
\end{figure}
\begin{figure}[htb]
\centering
\begin{subfigure}{0.8\linewidth}
\centering
\includegraphics[width=\linewidth]{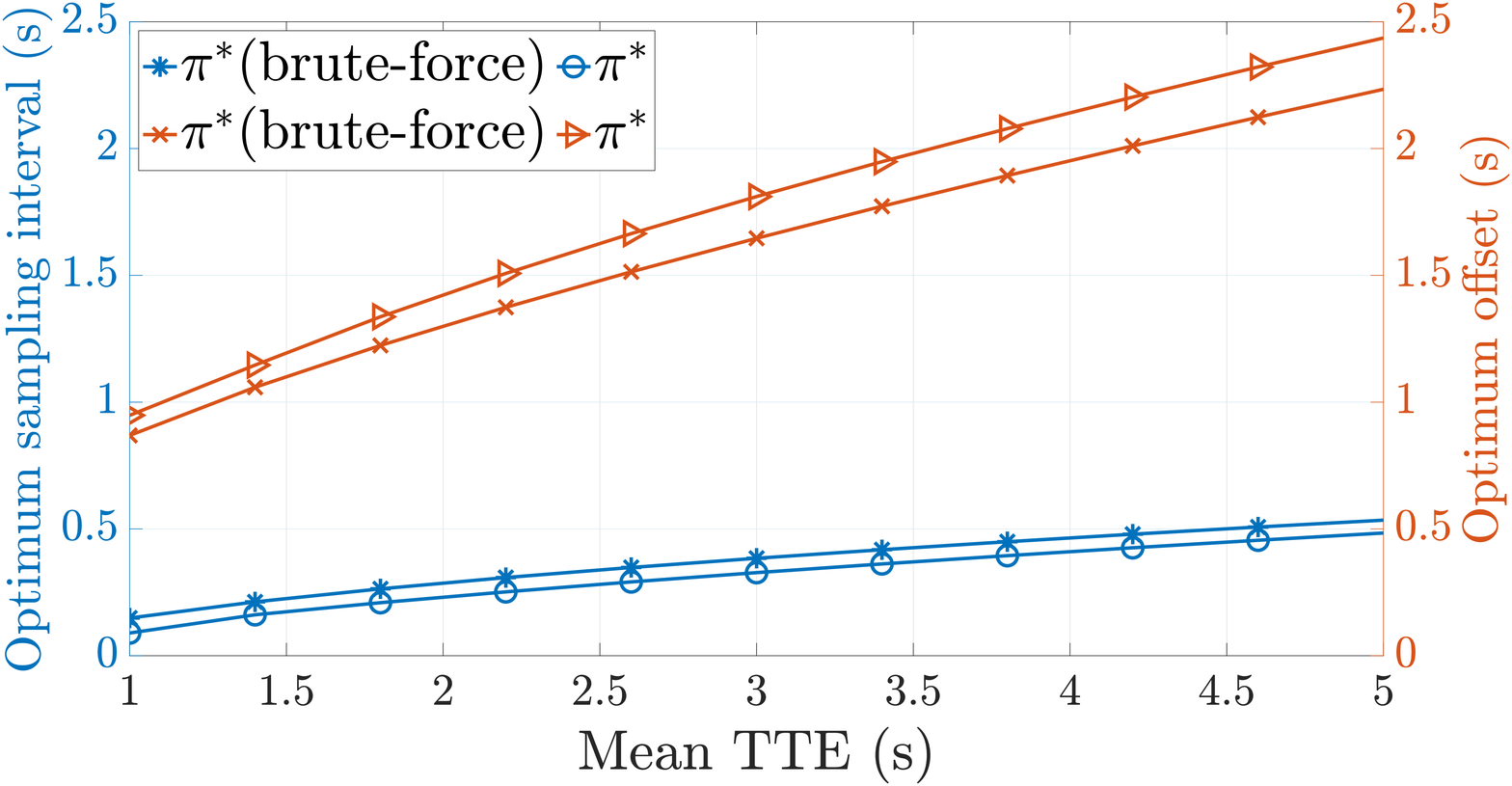}
\caption{$\ts^*,\;\delta^*$ vs. $\mathbb{E}[\mathcal{T}]$.}
\label{fig:bruteforceCOmparison_1}
\end{subfigure}
\begin{subfigure}{0.8\linewidth}
\centering
\includegraphics[width=\linewidth]{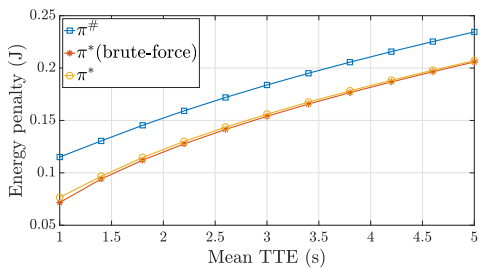}
\caption{$\mathcal{E}$ vs. $\mathbb{E}[\mathcal{T}]$.}
\label{fig:bruteforceCOmparison_2}
\end{subfigure}%
\caption{Optimum sampling interval $\ts^*$, optimum offset $\delta^*$ and the corresponding energy penalty $\mathcal{E}^*$ for various sampling policies vs. the mean TTE $\mathcal{E}[\mathcal{T}]$ for the VAS.}
\label{fig:bruteforceCOmparison}
\end{figure}
\subsection{Comparison of Sampling Policies}
We saw the battery life improvement provided by $\pi^*$ over $\pi_0$ earlier in Fig. \ref{fig:batterlife}. Even though we have two solutions -- $\pi^\#$ and $\pi^*$ -- for the VAS, we considered only the better performing $\pi^*$ in that figure.
Now in \ref{fig:bruteforceCOmparison}, we compare the performance of the two policies and see the benefit provided by adding the offset $\delta$. We also compare the proposed $\pi^*$ with the brute force solution discussed in Section \ref{prelim}. 
In Fig. \ref{fig:bruteforceCOmparison_1}, we compare the optimum sampling interval and offset obtained by Algorithm \ref{Algo} and brute-force, 
whereas Fig. \ref{fig:bruteforceCOmparison_2} compares the optimum energy penalty attained by $\pi^\#$ and $\pi^*$. From Fig. \ref{fig:bruteforceCOmparison_1}, we can see that there is an observable (albeit small) difference in $\ts^*$ (and $\delta^*$) obtained by the algorithm and the brute-force. Nonetheless, Fig. \ref{fig:bruteforceCOmparison_2} shows that the difference in energy penalty generated as a result of this difference is very minimal, thus confirming the near-optimality of Algorithm \ref{Algo}. We also observe that the additional penalty reduction offered by $\mathcal{P}$ over $\hat{\mathcal{P}}$ amounts to roughly $51\%$ at a mean task time of $1$s and $14\%$ at a mean task time of $5$s, thus demonstrating the advantage of introducing the offset.
\begin{figure}[t]
\centering
\begin{subfigure}{0.8\linewidth}
\centerline{\includegraphics[width=\linewidth]{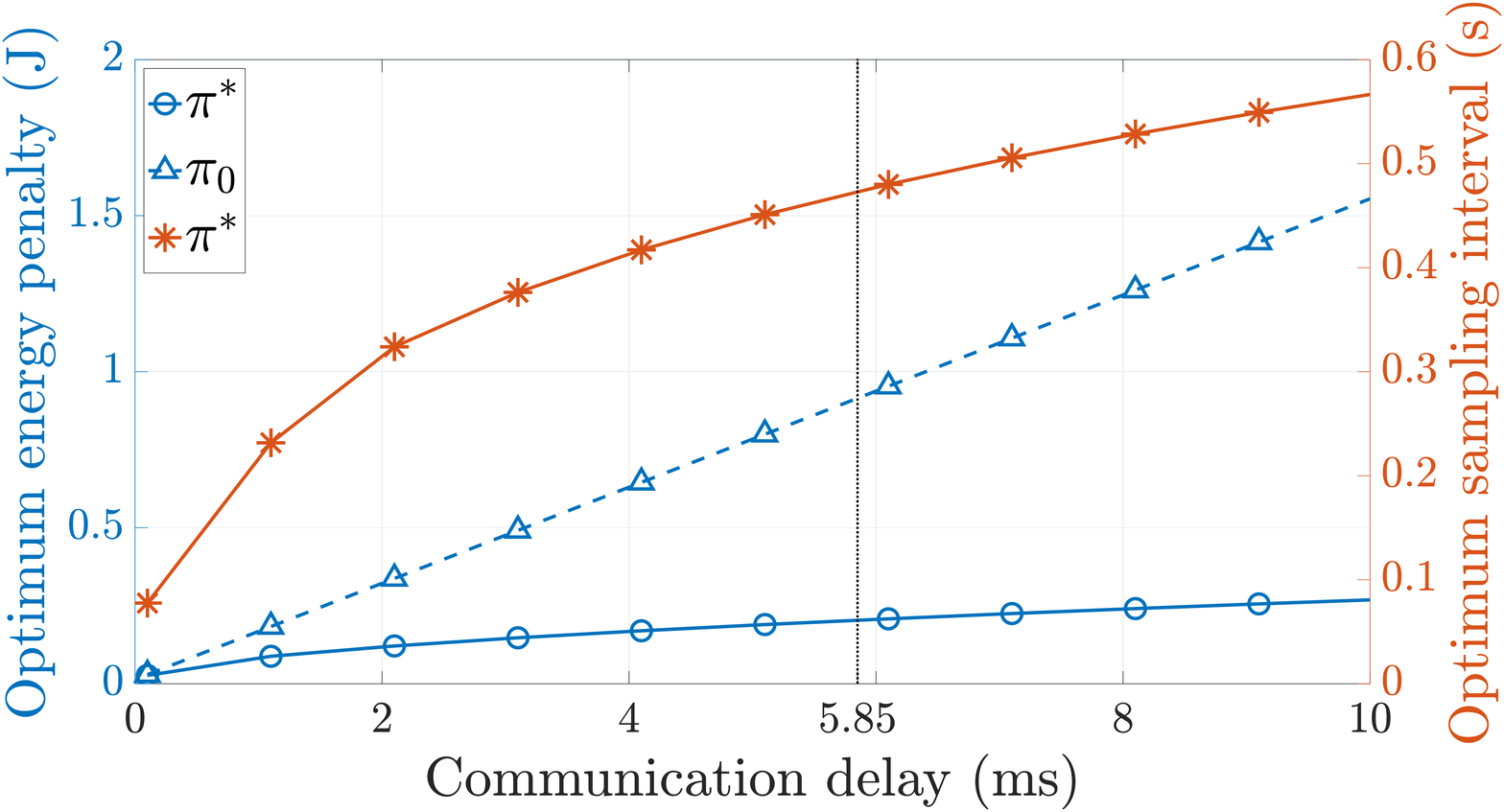}}
\caption{$\mathcal{E}^*,\,\ts^*$ vs. $\tau_\text{c}$ for fixed $P_\text{c},P_0$.}
\label{fig:r4_OptimumVsTau_c.eps}
\end{subfigure}
\begin{subfigure}{0.8\linewidth}
\centerline{\includegraphics[width=\linewidth]{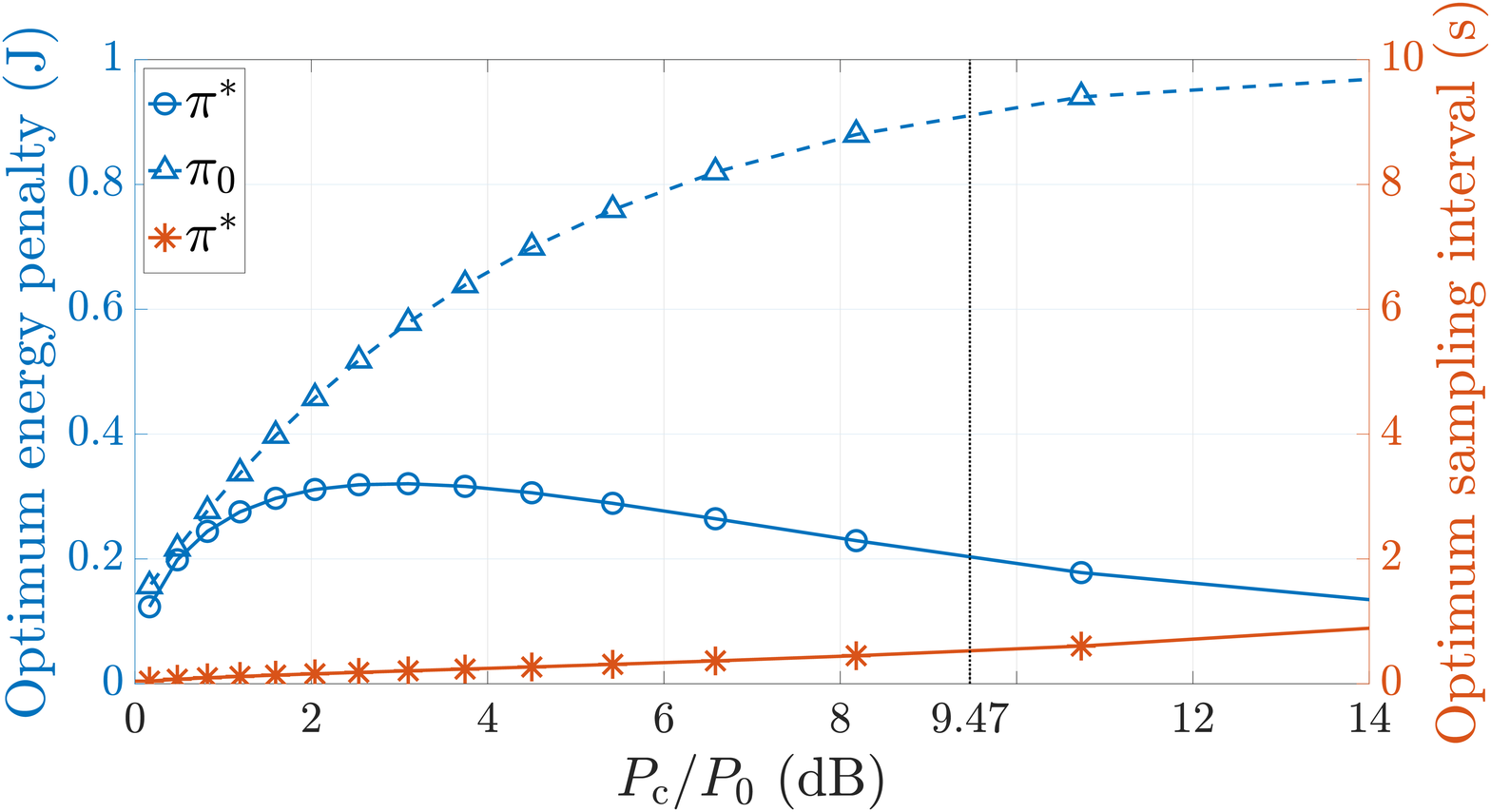}}
\caption{$\mathcal{E}^*,\,\ts^*$ vs. log$_{10}(\nicefrac{P_\text{c}}{P_0})$ for fixed $\tau_\text{c},P_\text{c}$.}
\label{fig:r6_OptimumVsP0.eps}
\end{subfigure}
\caption{Variation of optimum sampling interval $\ts^*$ and the corresponding energy penalty $\mathcal{E}^*$ for various sampling policies with variation in communication delay $\tau_\text{c}$ and idle power $P_0$ for the VAS. The vertical lines denotes $\tau_\mathrm{c}$ and log$_{10}(\nicefrac{P_\text{c}}{P_0})$ of the VAS from TABLE \ref{tab:Example}.}
\label{fig:rayleigh_optimumVsParams}
\end{figure}

\subsection{Variation of $\mathcal{E}$ with System Parameters}
Recall that the optimum sampling interval depends only on the ratio $\nicefrac{\beta}{\alpha}=\tau_{\text{c}}^{-1} (\tfrac{P_\text{c}}{P_0}-1)^{-1}$. Fig. \ref{fig:r4_OptimumVsTau_c.eps} and Fig. \ref{fig:r6_OptimumVsP0.eps} illustrates the variation of the energy penalty of the VAS with respect to the system parameters $\tau_\text{c}$ and $\nicefrac{P_\text{c}}{P_0}$. Observe from the figures that the energy penalty obtained by solving $\mathcal{P}$ varies much less with the system parameters than the baseline policy. This shows the relative independence of the optimum energy penalty with the system parameters.
Since these results are qualitatively similar for the CPS, we omit them due to space constraints. 

For fixed power, communication delay $\tau_\text{c}$ is dictated by the data size per sample and the transmission schemes. Fig. \ref{fig:r4_OptimumVsTau_c.eps} illustrate the dependency of $\mathcal{E}$ with respect to these aspects of communication. We can see that a larger communication delay results in larger $\ts^*$. However, the corresponding increase in the incurred penalty is much smaller for the proposed optimum. In Fig. \ref{fig:r6_OptimumVsP0.eps} we show the change in $\ts^*$ and $\mathcal{E}$ with respect to $\nicefrac{P_\text{c}}{P_0}$ for constant $\tau_\text{c}$ and ${P_\text{c}}$. We can see that the optimum sampling interval increases with an increase in the difference between the communication power and idle power. However, a more interesting observation is the corresponding change in the energy penalty. With $P_0<<P_\text{c}$, we obtain a significant penalty reduction by using $\pi^*$. On the other hand, as the $P_0$ approaches $P_\text{c}$ (i.e., 0dB in the figure), the penalty reduction goes to zero.

With advancements in semiconductor technology that reduces $P_0$ and the emergence of communication standards like the mmWave, $\nicefrac{P_\text{c}}{P_0}$ ratio is expected to increase from the 9.47dB that we modelled. Similarly, increasing image processing capabilities increase $\tau_\text{c}$ as well. As a result, it is evident from the figures that the gain offered by the proposed algorithm is more likely to increase in the future.

\section{Conclusion and Future Scope}\label{sec:conclusion}
We considered an edge-based feedback system that captures essential events via sampling and proposed an optimisation framework with which the sampling interval that minimises the energy consumption can be computed. Apart from the generic approach to solve the optimisation problem for an arbitrary task time distribution, we also considered two particular examples of interests -- a VAS (video analytics system) and a CPS (cyber physical system).
The TTE (time to event) of the VAS follows a Rayleigh distribution while that of the CPS follows an exponential distribution. These two systems are used to illustrate the behaviour of the variable components of energy. 
We discussed the energy savings enabled by the optimum sampling interval, the benefit provided by offsetting the first sample and the near-optimality of the proposed algorithm. We also discussed the dependency of the optimum to the system parameters like communication delay and idle power. Finally, we also discussed the additional energy expended as a result of a computation error and concluded that this expense is relatively steeper for a negative error.

From a mathematical perspective, the basic difference between a CPS and a VAS is the underlying TTE distribution. However, from a design perspective, the one important distinction is the benefit of adding the offset $\delta$. 
As explained before, offset is particularly useful when the statistical mode of the TTE distribution is large enough and the variance of the distribution is small enough so that the probability of events occurring in the early stages of a monitoring cycle is small.
As a result, offset is not useful for an application with an exponentially distributed TTE (see Proposition \ref{prop:OffsetIsIrrelevant}) but is useful when TTE is Rayleigh distributed. 
However, there are other distributions where the offset is much more useful and the proposed solution can provide a higher performance improvement. For instance, the exponentially modified Gaussian distribution can model the human response times in some similar applications. 

In this work, we considered the problem with only one single category of event. However in a CPS, systems might need to capture different categories of failures, and thus might have different severity or costs associated with them. Similarly in a VAS, the responsiveness of the system need to be considered and any minimisation of energy usage that affects the responsiveness over a certain limit should be avoided. Such demands for an additional characterisation that adds a variable weight to the required energy savings is the primary direction for the future scope of this work.

\bibliographystyle{IEEEtran}
\bibliography{00_main.bib}

%



\end{document}